\documentclass[reqno]{amsart}

\usepackage{amsmath, amsthm, amsfonts}
\usepackage[utf8]{inputenc}
\usepackage{enumerate}

\usepackage[colorlinks=true,linkcolor=blue,citecolor=blue]{hyperref}

\usepackage{enumerate}
\usepackage{mathrsfs}
\usepackage{amsopn}
\usepackage{amsmath}
\usepackage{amssymb}
\usepackage{amsfonts}
\usepackage{amsbsy}
\usepackage{amscd,indentfirst,epsfig}
\usepackage{amsfonts,amsmath,latexsym,amssymb,verbatim,amsbsy}
\usepackage{amsthm}
\numberwithin{equation}{section}

\newcommand{\ap}[1]{\left\langle#1\right\rangle}
\newcommand{\norm}[1]{\left\Vert#1\right\Vert}

\def\R{\mathbb{R}}

\def\Q{\mathcal{Q}}
\def\H{\mathcal{H} }

\def\W{W^L}
\def\L{\mathcal{L}}
\def \D {\mathscr{D}}
\def \leq {\leqslant}
\def \geq {\geqslant}

\def \d {\,\mathrm{d} }
\def \dt {\,\mathrm{d}t }
\def \ds {\,\mathrm{d}s }

\def \dy {\,\mathrm{d}y }

\def \du {\,\mathrm{d}u }
\def \dt {\,\mathrm{d}t }

\def \with {\ \big\vert \ }
\def \fa {\quad \text{ for all }}
\def \vartheta {\varpi}

\def \zs {z_{\mathrm{s}}}

\newcommand{\step}[1]{\medskip \noindent \textit{#1}}

\DeclareMathOperator{\sign}{sign}

\newtheorem{thm}{Theorem}[section]
\newtheorem{cor}[thm]{Corollary}
\newtheorem{lem}[thm]{Lemma}
\newtheorem{prp}[thm]{Proposition}
\newtheorem{hyp}[thm]{Hypothesis}
\theoremstyle{definition}
\newtheorem{defi}[thm]{Definition}
\theoremstyle{remark}
\newtheorem{rem}[thm]{Remark}
\theoremstyle{example}

\newtheorem{thmA}{Theorem A.}

\newtheorem{remA}[thmA]{Remark A.}

\newtheorem{thmB}{Theorem B.}
\newtheorem{lemB}[thmB]{Lemma B.}
\newtheorem{remB}[thmB]{Remark B.}

\title[Exponential convergence to equilibrium]{Exponential convergence
  to equilibrium for subcritical solutions of the Becker-D\"oring
  equations}

\date{January 2013}

\author{Jos\'e A. Ca\~nizo}

\address{\textbf{Jos\'{e} A.  Ca\~{n}izo}, School of Mathematics,
  University of Birmingham, Edgbaston, Birmingham B15 2TT, UK}
\email{j.a.canizo@bham.ac.uk} \thanks{JAC acknowledges support from
  the Marie-Curie CIG project KineticCF, the ERC grant MATKIT and the
  project MTM2011-27739-C04-02 from the Spanish \emph{Ministerio de
    Econom\'ia y Competitividad}. Part of this work was carried out
  during a visit of JAC to the University of Granada with the GENIL
  program.}

\author{Bertrand Lods}

\address{\textbf{Bertrand Lods}, Universit\`{a} degli Studi di Torino
  \& Collegio Carlo Alberto, Department of Statistics and Economics,
  Corso Unione Sovietica, 218/bis, 10134 Torino, Italy.}\thanks{BL
  acknowledges support from the \emph{de Castro Statistics
    Initiative}, Collegio Carlo Alberto, Moncalieri, Italy. This work
  started during a visit of BL to the Departamento de Matem\'{a}tica
  Aplicada, Universidad de Granada. Both authors would like to express
  their gratitude to M. J. C\'aceres for her kind hospitality.}
\email{lodsbe@gmail.com}

\begin{document}

\maketitle

\begin{abstract}
  We prove that any subcritical solution to the Becker-D\"{o}ring
  equations converges exponentially fast to the unique steady state
  with same mass. Our convergence result is quantitative and we show
  that the rate of exponential decay is governed by the spectral gap
  for the linearized equation, for which several bounds are
  provided. This improves the known convergence result by Jabin \&
  Niethammer \cite{JN03}. Our approach is based on a careful spectral
  analysis of the linearized Becker-D\"{o}ring equation (which is new
  to our knowledge) in both a Hilbert setting and in certain weighted
  $\ell^1$ spaces. This spectral analysis is then combined with
  uniform exponential moment bounds of solutions in order to obtain a
  convergence result for the nonlinear equation.\\

\noindent {\small \textsc{Keywords:}  Becker-D\"{o}ring equation, exponential trend to equilibrium, spectral gap, decay of semigroups.}
\end{abstract}

\tableofcontents

\section{Introduction}
\label{sec:intro}

\subsection{The Becker-D\"oring equations}

The Becker-D\"oring equations are a model for the kinetics of
first-order phase transitions, applicable to a wide variety of
phenomena such as crystallization, vapor condensation, aggregation of
lipids or phase separation in alloys. They give the time evolution of
the size distribution of clusters of a certain substance through the
following infinite system of ordinary differential equations:
\begin{subequations}
  \label{eq:BD}
  \begin{align}
    \label{eq:BD1}
    \frac{\d}{\d t} c_i(t) &= W_{i-1}(t) - W_{i}(t), \qquad i \geq 2,
    \\
    \label{eq:BD2}
    \frac{\d}{\d t} c_1(t) &= - W_1(t) - \sum_{k=1}^\infty W_{k}(t),
  \end{align}
\end{subequations}
where
\begin{equation}
  \label{eq:def-Wi}
  W_{i}(t) := a_{i}\, c_1(t) c_i(t) - b_{i+1}\, c_{i+1}(t)
  \qquad i \geq 1.
\end{equation}
Here the unknowns are the real functions $c_i = c_i(t)$ for $i \geq 1$
an integer, and represent the density of clusters of size $i$ at time
$t \geq 0$ (this is, clusters composed of $i$ individual particles).
They give the size distribution of clusters of the phase which is
assumed to have a small total concentration inside a large ambient
phase --- be it clusters of crystals, lipids, or droplets of water
forming in vapor. The numbers $a_{i}, b_{i}$ (for $i \geq 1$) are the
\emph{coagulation} and \emph{fragmentation coefficients},
respectively, and we always assume them to be strictly positive, i.e.,
\begin{equation}
  \label{eq:hyp-positivity}
  a_i > 0 \quad \text{ for integer $i \geq 1$},
  \qquad b_i > 0 \quad \text{ for integer $i \geq 2$}.
\end{equation}
We set $b_1 = 0$ for convenience in writing some of the equations. If
we represent symbolically by $\{i\}$ the species of clusters of size
$i$, then \eqref{eq:BD} is based on the assumption that the only
(relevant) chemical reactions that take place are
\begin{equation*}
  \{i\} + \{1\} \rightleftharpoons \{i+1\}
\end{equation*}
for integer $i \geq 1$. These reactions all take place both ways; the
rate at which $\{i\}+\{1\} \rightarrow \{i+1\}$ occurs is $a_i c_1
c_i$, while the rate at which $\{i+1\} \rightarrow \{i\} + \{1\}$
occurs is $b_{i+1} c_{i+1}$. These rates are proportional to the
concentrations of the species appearing in the left hand side, in
agreement with the law of mass action from chemistry. Note that $b_i$
appears in the system \eqref{eq:BD} only for $i \geq 2$, since single
particles cannot break any further. The quantity $W_i(t)$ defined in
\eqref{eq:def-Wi} thus represents the \emph{net rate} of the reaction
$\{i\} + \{1\} \rightleftharpoons \{i+1\}$. The sum
\begin{equation}
  \label{eq:def-mass}
  \varrho :=
  \sum_{i=1}^\infty i c_i(0)
  = \sum_{i=1}^\infty i c_i(t)
  \quad \text{ for all } t \geq 0.
\end{equation}
is usually called the \emph{density} or \emph{mass} of the solution,
and is a conserved quantity of the evolution.

The system \eqref{eq:BD} was originally proposed in \cite{BD35},
though with a variation: the monomer density $c_1$ was assumed
constant and hence equation \eqref{eq:BD2} did not appear. The present
version, where the total density is conserved, was first discussed in
\cite{B77} and \cite{PL79}, and is a widely used model in classical
nucleation theory. We refer to these works, as well as the more recent
reviews \cite{Schmelzer05,Oxtoby} for a background on the physics and
applications of the Becker-D\"oring equations.

We define the \emph{detailed balance coefficients} $Q_i$ recursively
by
\begin{equation}
  \label{eq:Qi}
  Q_1 = 1,
  \quad
  Q_{i+1} = \frac{a_i}{b_{i+1}} Q_i
  \quad \text{ for } i \geq 1.
\end{equation}
An \emph{equilibrium} of equation \eqref{eq:BD} is a constant-in-time
solution (with finite mass). For a given $z \geq 0$, the sequence $c_i
:= Q_i z^i$ for $i \geq 1$ is formally an equilibrium of
\eqref{eq:BD}, since all of the $W_i$ vanish. However, some of these
sequences do not have finite mass (and hence are not equilibria). The
largest possible number $\zs \geq 0$ (possibly $+\infty$) for which
$\sum_i i Q_i z^i < +\infty$ for all $0 \leq z < \zs$ is called the
\emph{critical monomer density} (i.e., $\zs$ is the radius of
convergence of the power series with coefficients $i Q_i$). The
quantity $z_\mathrm{s}$ is also called monomer saturation density,
hence the subscript. The \emph{critical mass} (or, again, saturation
mass) is then defined by
\begin{equation}
  \label{eq:critical-mass}
  \varrho_\mathrm{s} := \sum_{i=1}^\infty i Q_i z_\mathrm{s}^i
  \in [0,+\infty].
\end{equation}
We emphasize that both $\varrho_\mathrm{s}$ and $z_\mathrm{s}$ are
completely determined by the coefficients $a_i$, $b_i$. It is clear
from this that the sequences $\{Q_i z^i\}$ are equilibria for $z <
z_\mathrm{s}$, and $\{Q_i z_\mathrm{s}^i\}$ is also an equilibrium if
$\varrho_\mathrm{s} < +\infty$. These are in fact the only finite-mass
equilibria \cite{BCP86}. Since $z \mapsto \sum_{i=1}^\infty i Q_i z^i$
is continuous and strictly increasing for $z < z_\mathrm{s}$ (and up
to $z = z_\mathrm{s}$ if $\varrho_\mathrm{s} < +\infty$), one sees
that for any finite mass $\varrho \leq \varrho_\mathrm{s}$ there is a
unique equilibrium with mass $\varrho$. We call a solution
\emph{subcritical} when it has mass $\varrho < \varrho_\mathrm{s}$,
\emph{critical} when its mass is exactly $\varrho_\mathrm{s}$, and
\emph{supercritical} when it has mass $\varrho > \varrho_\mathrm{s}$
(the two latter cases only make sense provided $\varrho_\mathrm{s} <
+\infty$).

From the above discussion one sees that for a subcritical or critical
solution there exists an equilibrium with the same mass as the
solution; this is, there exists $z \leq \zs$ (which must be unique)
for which
\begin{equation}
  \label{eq:z-def}
  \sum_{i=1}^\infty i Q_i z^i = \sum_{i=1}^\infty i c_i(t)
  \left( = \sum_{i=1}^\infty i c_i(0) = \varrho\right)
  \quad \text{ for } t \geq 0.
\end{equation}
When talking about a subcritical or critical solution we will often
denote by $z$ the unique number satisfying \eqref{eq:z-def}, and refer
to it as the \emph{equilibrium monomer density}. On the other hand,
for supercritical solutions there is no equilibrium with the same mass
as the solution.

The critical density $\varrho_\mathrm{s}$ marks a difference in the
behavior of solutions: above the critical density a phase transition
phenomenon takes place, reflected in the fact that the excess density
$\varrho-\varrho_\mathrm{s}$ is concentrated in larger and larger clusters
as time passes. On the other hand, below or at the critical density a
stationary state of the same mass as the solution is eventually
reached. Since we are concerned with the study of subcritical
solutions we will always assume that
\begin{equation}
  \label{eq:critical-density-positive}
  z_\mathrm{s} > 0
  \qquad
  \text{(equivalently, $\varrho_\mathrm{s} > 0$).}
\end{equation}

A fundamental quantity related to \eqref{eq:BD} is the \emph{free
  energy} $H(c)$, defined (at least formally) for any sequence $c =
(c_i)_{i \geq 1}$ by
\begin{equation}
  \label{eq:def-V}
  H(c) := \sum_{i=1}^\infty c_i \left( \log \frac{c_i}{Q_i} - 1 \right),
\end{equation}
which decreases along solutions of \eqref{eq:BD} (see \cite{BCP86}):
for a (strictly positive, suitably decaying for large $i$) solution $c
= c(t) = (c_i(t))_{i \geq 1}$ of \eqref{eq:BD} we have
\begin{equation}
  \label{eq:H-thm-BD}
  \frac{\d}{\dt} H(c(t))
  = - D(c(t))
  := -\sum_{i=1}^\infty a_i Q_i
  \left(\frac{c_1 c_i}{Q_i} - \frac{c_{i+1}}{Q_{i+1}} \right)
  \left(\log \frac{c_1 c_i}{Q_i}
    - \log \frac{c_{i+1}}{Q_{i+1}} \right)
  \leq 0
\end{equation}
for $t \geq 0$. Since mass is conserved, if one fixes $z \leq z_s$
then the same identity is true of the functional $F_z = F_z(c)$ given
by
\begin{equation}
  \label{eq:def-VFz}
  F_z(c) := \sum_{i=1}^\infty c_i \left( \log \frac{c_i}{z^i Q_i} - 1
  \right) + \sum_{i=1}^\infty z^i Q_i
  = H(c) - \log z \sum_{i=1}^\infty i c_i
  + \sum_{i=1}^\infty z^i Q_i,
\end{equation}
this is,
\begin{equation}
  \label{eq:H-thm-BD-Fz}
  \frac{\d}{\dt} F_z(c(t))
  = - D(c(t))
  \qquad \text{ for } t \geq 0.
\end{equation}
It is readily checked that $F_z(z^i Q_i) = 0$, i.e., $F_z$ vanishes at
the equilibrium $(z^i Q_i)_{i \geq 1}$.

\subsection{Typical coefficients}

Remarkable model coefficients appearing in
the theory of density-conserving phase transitions (see
\cite{Penrose89,N08}) are given by
\begin{equation}
  \label{eq:PT-coefficients}
  a_i=i^\alpha,
  \qquad
  b_i = a_i \left(z_\mathrm{s}+\dfrac{q}{i^{1-\mu}}\right)
  \fa i \geq 1,
\end{equation}
for some $0 < \alpha \leq 1$, $z_\mathrm{s} >0$, $q >0$ and $0 <
\mu < 1$. These coefficients may be derived from simple assumptions
on the mechanism of the reactions taking place; we take particular
values from \cite{N08}:
\begin{equation}
  \label{eq:example-values}
  \begin{split}
    \alpha = 1/3, \quad \mu = 2/3 &\qquad \text{(diffusion-limited
      kinetics in 3-D),}
    \\
    \alpha = 0, \quad \mu = 1/2 &\qquad \text{(diffusion-limited
      kinetics in 2-D),}
    \\
    \alpha = 2/3, \quad \mu = 2/3 &\qquad
    \text{(interface-reaction-limited kinetics in 3-D),}
    \\
    \alpha = 1/2, \quad \mu = 1/2 &\qquad
    \text{(interface-reaction-limited kinetics in 2-D).}
  \end{split}
\end{equation}
One obtains from \eqref{eq:PT-coefficients} and
\eqref{eq:Qi} that in this case
\begin{equation*}
  Q_i
  = \dfrac{a_1 a_2 \ldots a_{i-1}}{b_2 b_3 \ldots b_i}
  = i^{-\alpha } \prod_{j=2}^{i}
  \left( z_\mathrm{s} + \frac{q}{j^{1-\mu}} \right)^{-1}
  \fa \: i \geq 1,
\end{equation*}
where the product is understood to be equal to $1$ for $i=1$ (so $Q_1
= 1$). One can deduce from this that the critical monomer density is
indeed the quantity $z_\mathrm{s}$ appearing in
\eqref{eq:PT-coefficients}, and that $\varrho_\mathrm{s}$ is a finite
positive quantity.

Another kind of reasoning that leads to a similar set of coefficients
is the following: while one may determine the coagulation coefficients
through an understanding of the aggregation mechanism and estimate
that
\begin{equation}
  \label{eq:CF-coefficients1}
  a_i = i^\alpha
\end{equation}
for some $0 < \alpha \leq 1$ as above, the equilibrium state of a
system may be obtained from statistical mechanics considerations: it
should be
\begin{equation*}
  c_i = c_1^i \exp\left(-\epsilon_i\right),
\end{equation*}
where $\epsilon_i$ is the binding energy of a cluster of size $i$ (the
energy required to assemble a cluster of size $i$ from $i$ monomers,
in units of Boltzmann's constant times the temperature). Depending on
the kind of aggregates considered, this binding energy can be
estimated (at least for large $i$) as
\begin{equation*}
  \epsilon_i = -\beta (i-1) + \sigma (i-1)^{\mu},
\end{equation*}
where $\beta$ is the energy released when adding one particle to a
cluster, and $\sigma$ is related to the
surface tension of the aggregates. The values of $\mu$ and $\alpha$
for various situations are still those in
\eqref{eq:example-values}. From this we may deduce that
\begin{equation}
  \label{eq:CF-Qi}
  Q_i = \exp\left(\beta(i-1) - \sigma (i-1)^{\mu}\right)
  = z_\mathrm{s}^{1-i} \exp\left(-\sigma(i-1)^\mu\right),
\end{equation}
where we define $z_\mathrm{s} := \exp\left(-\beta\right)$. By \eqref{eq:Qi} and
\eqref{eq:CF-coefficients1} we finally have
\begin{equation}
  \label{eq:CF-coefficients}
  a_i=i^\alpha,
  \qquad
  b_i = \zs (i-1)^\alpha
  \exp\big(\sigma i^\mu - \sigma (i-1)^\mu\big)
  \fa i \geq 1,
\end{equation}
for some $0 \leq \alpha \leq 1$, $\zs > 0$ and $0 < \mu <
1$. In this case $z_\mathrm{s}$ is still consistent with our
definition of the critical monomer density. Observe that for large $i$
we have $i^\mu - (i-1)^\mu \sim \mu i^{\mu-1}$, so the fragmentation
coefficients become roughly
\begin{equation}\label{eq:CF-equivbi}
  b_i \sim \zs\, a_i
  \exp\big(\sigma \mu i^{\mu-1} \big)
  \sim
  a_i
  \left( z_\mathrm{s} + \frac{z_\mathrm{s} \sigma \mu}{i^{1-\mu}} \right),
\end{equation}
which is like \eqref{eq:PT-coefficients} with $q = z_\mathrm{s} \sigma
\mu$.

Both example coefficients \eqref{eq:CF-coefficients} and
\eqref{eq:PT-coefficients} are often used in the literature, so we
always state explicitly the consequences of our results for them.

\subsection{Main results}

The mathematical theory of the Becker-Döring equations \eqref{eq:BD}
has been developed in a number of papers since the first rigorous
works on the topic \cite{BCP86,BC88}. Existence and uniqueness of a
solution for all times is ensured \cite[Theorems 2.2 and 3.6]{BCP86}
provided that
\begin{equation}
  \label{eq:hyp-existence}
  a_i = O(i),
  \quad
  b_i = O(i),
  \qquad
  \sum_{i=1}^\infty i^2 c_i(0) < +\infty.
\end{equation}
As expected, the unique solution conserves mass (this is,
\eqref{eq:def-mass} holds rigorously). Weaker, or different,
conditions may be imposed for this to hold, but
\eqref{eq:hyp-existence} will be enough for our purposes in the
present paper. Since we use some results from \cite{JN03}, we often
also require the following:
\begin{subequations}
  \label{eq:hyp-JN}
  \begin{gather}
    \label{eq:ai-bounded-below}
    \inf_{k \geq 1} a_k =: \underline{a} > 0,
    \\
    \label{eq:JN-zs-finite}
    0 < z_\mathrm{s} < +\infty,
    \\
    \label{eq:Qi-limit}
    \lim_{i \to +\infty} \frac{Q_{i+1}}{Q_{i}} = \frac{1}{\zs},
    \\
    \label{eq:ai-limit}
    \lim_{i \to +\infty} \dfrac{a_{i+1}}{a_i} = 1.
  \end{gather}
\end{subequations}
Regarding \eqref{eq:Qi-limit}, we recall that $\zs$ is the radius of
convergence of the power series with coefficients $i \Q_i$; hence, if
the limit of $Q_{i} / Q_{i+1}$ exists, it can only be $\zs$. We notice
that the conditions in \cite{JN03} are slightly different from the
ones above. However, the differences are not essential and all
arguments can be carried out with the ones here, which are more
convenient and are satisfied by the example coefficients
\eqref{eq:PT-coefficients} and \eqref{eq:CF-coefficients}. Our study
of the spectral properties of the linearized operator, and in
particular the proof that it is self-adjoint in Proposition
\ref{prp:L-decomposition}, seems to require the above limiting
behavior of the ratios $Q_{i+1}/Q_i$ and $a_{i+1}/a_i$ (compare to the
conditions in \cite{JN03}). It may be possible to modify our proofs to
accommodate slightly more general conditions, but we avoid this in
order to simplify the main argument.

The relationship between the long-time behavior of
supercritical solutions and late-stage coarsening theories is
especially interesting, and has been studied in
\cite{Pen97,Vel98,CGPV02,N08}; in particular, in \cite{N08} it was
rigorously shown that a suitable rescaling of supercritical solutions
approximates a solution of the Lifshitz-Slyozov equations \cite{LS61}.
Asymptotic approximations of solutions have been developed in
\cite{NCB02,citeulike:1069570,citeulike:2862238,citeulike:4140416}.

For subcritical solutions (the regime which we study in the present
paper) it was proved in \cite{BCP86,BC88} that a stationary state with
the same mass as the solution is approached as time passes. The rate
of this convergence was studied in \cite{JN03} and was found to be at
least like $\exp(-C t^{1/3})$ for some constant $C$. Our purpose here
is to improve this by showing that in fact the speed of convergence is
exponential (this is, like $\exp(-\lambda t)$ for some $\lambda >
0$). We give upper and lower bounds of the constant $\lambda$, and
study its behavior as the density $\varrho$ of the solution approaches
(from below) the critical density $\varrho_\mathrm{s}$. An important
quantity which appears in our analysis is
\begin{equation}
  \label{eq:def-B}
  B := \sup_{k \geq 1} \left(\sum_{j=k+1}^\infty Q_j z^j \right)
  \left(\sum_{j=1}^{k}\dfrac{1}{a_j Q_j z^j}\right) \in (0,+\infty].
\end{equation}
Observe that $B$ depends only on the coefficients $a_i$ ($i \geq 1$)
and $b_i$ ($i \geq 2$). We are able to show that subcritical solutions
of \eqref{eq:BD} converge exponentially to equilibrium and $1/B$ is
a good measure of the speed of convergence whenever $B$ is finite:
\begin{thm}
  \label{thm:main-intro*}
  Assume \eqref{eq:hyp-existence} and \eqref{eq:hyp-JN}. Let $c =
  (c_i)_i$ be a nonzero subcritical solution to equation \eqref{eq:BD}
  with initial condition $c(0)$ such that, for some $\nu > 0$,
  \begin{equation}
    \label{eq:initial-exp-moment-finite}
    M := \sum_{i=1}^\infty \exp\left(\nu i\right) c_i(0) < +\infty,
  \end{equation}
  and define the equilibrium monomer density $z > 0$ by
  \eqref{eq:z-def}. Then there exist explicit $\overline{\nu} \in
  (0,\nu)$ and $\lambda_\star > 0$ such that, for any $\eta \in
  (0,\overline{\nu})$, there is a number $C > 0$ which depends only on
  $\varrho$, $\eta$, $M$ and $\sum_{i=1}^\infty e^{\eta i} z^i Q_i$ such
  that
  \begin{equation}
    \label{eq:main-convergence-expo}
    \sum_{i=1}^\infty \exp\left(\eta i\right) \left|c_i(t) - z^i Q_i\right|
    \leq C \exp\left(-\lambda_\star t\right)
    \fa t \geq 0,
  \end{equation}
  Under these conditions $B$ is finite, and if $\lim_{i \to +\infty}
  a_i = +\infty$ we may take $\lambda_\star = 1/B$.
\end{thm}
For the particular case of the example coefficients
\eqref{eq:PT-coefficients} or \eqref{eq:CF-coefficients}, if $\alpha <
2(1-\mu)$ we show that $B$ can be bounded by
\begin{equation*}
  C_1 \left(\log\frac{z_\mathrm{s}}{z}\right)^{-2 + \frac{\alpha}{1-\mu}}
  \leq
  B
  \leq
  C_2 \left(\log\frac{z_\mathrm{s}}{z}\right)^{-2 + \frac{\alpha}{1-\mu}}
\end{equation*}
for $z < z_\mathrm{s}$, where $C_1$ and $C_2$ depend only on the
coefficients $(a_i)_i$, $(b_i)_i$ (and in particular are independent
of $z$). In the case $\alpha \geq 2(1-\mu)$ we have
\begin{equation*}
  C_1
  \leq
  B
  \leq
  C_2
\end{equation*}
for $z \leq z_\mathrm{s}$ and some (other) constants $C_1$, $C_2$ that
again depend only on the coefficients. We observe that if $z =
z_\mathrm{s}$ then $B$ is finite if and only if $\alpha \geq
2(1-\mu)$.

\subsection{Method of proof}

Our proof is based on a study of the linearization of the
Becker-Döring equations around the equilibrium, for which we show the
existence of a spectral gap whose size is well estimated by $1/B$ (see
Theorem \ref{thm:spectral-gap} for details). This implies exponential
convergence to equilibrium for the linearized system, which can be
extended to the nonlinear equations by means of techniques developed
in the literature on kinetic equations, and particularly on the
Boltzmann equation \cite{M06,GMM}.

We observe that the improvement with respect to \cite{JN03} comes from
the use of a different method. The main tool in \cite{JN03} is an
inequality between the free energy (or entropy) $H$ defined by
\eqref{eq:def-V} and its production rate $D$ (see
\eqref{eq:H-thm-BD}) in the spirit of the ones available for the
Boltzmann equation \cite{DMV11}. As pointed out in \cite{JN03}, an
inequality like $H \leq C D$ for some constant $C > 0$, which
would directly imply an exponential convergence to equilibrium, is
roughly analogous to a functional log-Sobolev inequality, which is
known \emph{not} to hold for a measure with an exponential tail.
Since this is the case for the stationary solutions of the
Becker-Döring equation, it is believed (though, to our knowledge, not
proved) that this inequality does not hold in general for this
equation; hence, the following weaker inequality (this is, weaker for
small $H$) is proved in \cite{JN03}:
\begin{equation*}
  \frac{H}{|\log H|^2} \leq C D,
\end{equation*}
implying a convergence like $\exp(-C t^{1/3})$. This obstacle has a
parallel in the Boltzmann equation, for which the corresponding
inequality (known as Cercignani's conjecture) has been proved not to
hold in general, and can be substituted by inequalities like
$H^{1+\epsilon} \leq C D$ for $\epsilon > 0$ (we refer to the recent
review \cite{DMV11} for the history of the conjecture and a detailed
bibliography). However, just as for the space homogeneous \cite{M06}
and the full Boltzmann equation \cite{GMM} this can be complemented by
the study of the linearized equation in order to show full exponential
convergence. By following a parallel reasoning for the Becker-Döring
system we can upgrade the convergence rate to exponential.\medskip

Hence, our analysis is built around a study of the linearized
Becker-D\"{o}ring equation, which is new to our knowledge. We prove
here the existence of a positive spectral gap of the operator
$\mathbf{L}$, defined in Section \ref{sec:linear} as a suitable
linearization of Eq. \eqref{eq:BD} around the equilibrium $(\Q_i)_{i
  \geq 1}=(z^i Q_i)_{i \geq 1}$, in different spaces:

\begin{enumerate}
\item We provide first a spectral description of $\mathbf{L}$ in a
  Hilbert space setting. Namely, we shall investigate the spectral
  properties of the operator $\mathbf{L}$ in the weighted space
  $\mathcal{H}=\ell^2(\Q).$ This analysis is carried out with two
  (complementary) techniques: on the one hand, under reasonable
  conditions on the coefficients, one can show that $\mathbf{L}$ is
  self-adjoint in $\mathcal{H}$ and, resorting to a compactness
  argument, the existence of a non constructive spectral gap can be
  shown. On the other hand, using a discrete version of the weighted
  Hardy's inequality, the positivity of the spectral gap is completely
  characterized in terms of necessary and sufficient conditions on the
  coefficients. Moreover, and more importantly, quantitative estimates
  of this spectral gap are
  given.

\item Unfortunately, as it occurs classically for kinetic models, the
  Hilbert space setting which provides good estimates for the
  linearized equation is usually not suitable for the nonlinear
  equation. Thus, inspired by previous results on Navier-Stokes and
  Boltzmann equation \cite{gallay, M06}, we derive the spectral
  properties of the linearized operator in a larger weighted $\ell^1$
  space. We use for this an abstract result (see Theorem
  \ref{thm:GMM}) allowing to enlarge the functional space in which the
  exponential decay of a semigroup holds. This follows recent
  techniques developed in \cite{GMM}, though we give a self-contained
  proof simplified in our setting. The application of this theoretical
  result requires some important technical efforts, see Theorem
  \ref{thm:spectral-gap-extension}.
\end{enumerate}

It is worth pointing out that our techniques parallel the historical
development of the study of the exponential decay of the homogeneous
Boltzmann equation. We first show by non-constructive methods based on
Weyl's theorem that the linearized operator $\mathbf{L}$ has a
positive spectral gap. An exposition of similar techniques for the
linearized homogeneous Boltzmann equation can be found in
\cite{CIP}. Explicit estimates for this spectral gap were given in
\cite{BM05}, and similar techniques for the extension of the spectral
gap were devised in \cite{M06} and developed in \cite{GMM}, and used
to study the rate of convergence to equilibrium of the homogeneous
Boltzmann equation.

In Section \ref{sec:linear} we carry out the plan in point (1) above,
and in Section \ref{sec:spectralL1} we carry out point (2). The
application to the nonlinear equation and the proofs of our main
results are then given in Section \ref{sec:nonlinear}.

\section{The linearized Becker-Döring equations}
\label{sec:linear}

\subsection{The linearized operator}
\label{sec:operator}

For the whole of Section \ref{sec:linear} we assume the following:

\begin{hyp}
  \label{hyp:linear-section}
  We take $(a_{i})_{i \geq 1}$ and $(b_{i})_{i \geq 1}$ satisfying
  \eqref{eq:hyp-positivity}, define $(Q_i)_{i \geq 1}$ by
  \eqref{eq:Qi}, and assume \eqref{eq:critical-density-positive}. We
  also take $0 < z \leq \zs$ and set
  \begin{equation}
    \label{eq:scriptQi-def}
    \Q_i = Q_i z^i, \qquad i \geq 1.
  \end{equation}
  We also assume that
  \begin{equation}
    \label{eq:A}
    A := \sum_{i=1}^\infty i^2 (1 + a_i + b_i)^2 \Q_i < +\infty.
  \end{equation}
\end{hyp}
We remark that, when $z < \zs$, the sum in \eqref{eq:A} is indeed
finite under ``reasonable'' conditions on the coefficients (such as
\eqref{eq:hyp-existence}). Hence condition \eqref{eq:A} is important
mainly for the case $z = \zs$. Also, when $z = \zs$, condition
\eqref{eq:A} ensures that $\varrho_\mathrm{s} < +\infty$.\medskip

The choice of $0 < z \leq \zs$ corresponds to the choice of a mass $0
< \varrho \leq \varrho_\mathrm{s}$ given by \eqref{eq:z-def}. The unique
equilibrium with mass $\varrho$ is precisely $(\Q_i)_{i \geq 1}$, given
by \eqref{eq:scriptQi-def}. Notice that \eqref{eq:Qi} implies that
\begin{equation}
  \label{eq:DB-Qizi}
  a_i \Q_1 \Q_i = b_{i+1} \Q_{i+1},
  \qquad i \geq 1.
\end{equation}
Consider a solution $(c_i)_{i \geq 1}$ of \eqref{eq:BD} with mass
$\varrho$ (this is, a subcritical or critical solution). In order to
linearize equation \eqref{eq:BD} around the steady state $(\Q_i)_i$ we
define the fluctuation $h = (h_i)_{i \geq 1}$ by
\begin{equation}
  \label{eq:def-hi}
  c_i =\Q_i\left( 1+h_i \right)
\end{equation}
where the components of $h$ may have any sign. Let us carry out some
formal computations, and leave the precise definition of the
linearized operator for Section \ref{sec:linearized}. Notice that, in
order for \eqref{eq:z-def} to be satisfied, it is necessary that
\begin{equation}
  \label{ortho}
  \sum_{i=1}^\infty i \Q_i h_i(t) =0 \qquad \text{ for all } t \geq 0.
\end{equation}
The weak form of \eqref{eq:BD} reads as follows:
\begin{equation}
  \label{weak}
   \sum_{i=1}^\infty \phi_i \frac{\d}{\d t} c_i
  = \sum_{i=1}^\infty
  \big( a_{i} c_i(t) c_1(t)- b_{i+1} c_{i+1}(t) \big)
  \,\big( \phi_{i+1} - \phi_i - \phi_1 \big)
\end{equation}
for any sequence $(\phi_i)_{i \geq 1}$. Plugging into this the ansatz
\eqref{eq:def-hi} yields, for any sequence $(\phi_i)_i$
\begin{multline}
   \sum_{i=1}^\infty\dfrac{\d}{\d t}h_i(t)\Q_i \phi_i
  \\
  =  \sum_{i=1}^\infty \bigg( a_{i}\Q_i \Q_1
  (1+h_i(t))\,(1+h_1(t))-b_{i+1}\Q_{i+1}(1+h_{i+1}(t))\bigg)
  \,\bigg(\phi_{i+1}-\phi_i-\phi_1\bigg).
\end{multline}
Using \eqref{eq:DB-Qizi} one sees that, for any $i,j \geq 1$,
\begin{equation*}
  \begin{split}
    a_{i}\Q_i \Q_1
  (1+h_i(t))\,&(1+h_1(t))-b_{i+1}\Q_{i+1}(1+h_{i+1}(t))\\
    &=a_{i}\Q_i \Q_j \left(h_i(t)+h_1(t)-h_{i+1}(t)\right) + a_{i}\Q_i \Q_1  h_i(t)\,h_1(t).
  \end{split}
\end{equation*}
This means that the fluctuation $h(t)=(h_i(t))_i$ satisfies
\begin{equation}
  \label{eq:hi-equation}
  \dfrac{\d }{\d t}h_i(t)
  =
  L_i(h(t)) + \Gamma_i(h(t),h(t))
\end{equation}
where the linear operator $L$ is given, in weak
form, by
\begin{equation}
  \label{eq:weak-linear}
  \begin{split}
    \sum_{i=1}^\infty L_i(h)\Q_i\,\phi_i
    &=
    \sum_{i\geq 1}
    \W_{i}\,
    \big(\phi_{i+1}-\phi_i-\phi_1\big)
    \\
    &= \sum_{i=1}^\infty a_{i}\Q_i \Q_1 \big( h_i +
    h_1 - h_{i+1} \big) \big(\phi_{i+1}-\phi_i-\phi_1\big)
  \end{split}
\end{equation}
for any sequences $h=(h_i)_i$, $(\phi_i)_i$, where
\begin{equation}
  \label{eq:def-Wij-linearized}
  \W_{i} := a_{i} \Q_i \Q_1 (h_i + h_1) - b_{i+1} \Q_{i+1} h_{i+1}
  = a_{i} \Q_i \Q_1 (h_i + h_1 - h_{i+1}).
\end{equation}
Alternatively, we may write $L$ in strong form as
\begin{equation}
  \label{eq:L}
  L_1(h) =
  -\frac{1}{\Q_1}\left(\W_1+ \sum_{k=1}^\infty\W_k\right),
  \qquad
  L_i(h) =
  \frac{1}{\Q_i} \left(\W_{i-1} - \W_{i}\right)
  \qquad (i \geq 2).
\end{equation}
The bilinear operator $\Gamma(f,g)$ is defined in weak form by
\begin{equation*}
  \sum_{i=1}^\infty \Gamma_i(f,g) \Q_i \phi_i
  = \frac{1}{2}\sum_{i \geq 1} a_{i}\Q_i \Q_1 \left(f_i\,g_1+f_1\,g_i\right)
  \big(\phi_{i+1}-\phi_i-\phi_1\big)
\end{equation*}
for any sequences $f=(f_i)_i$, $g=(g_i)_i$ and
$(\phi_i)_i$. Alternatively, the strong form of $\Gamma_i(f,g)$
can be written as
\begin{equation}
  \label{eq:def-Gamma1}
  \Gamma_1(f,g)
  =
  -a_1\Q_1 f_1\,g_1
  - \frac{1}{2}\sum_{i=1}^\infty a_i \Q_i
  \left( f_i\,g_1 + f_1\,g_1 \right)
\end{equation}
and
\begin{equation}
  \label{eq:def-Gamma}
  \Gamma_i(f,g)
  =
  \frac{1}{2 \Q_i}\bigg(a_{i-1}\Q_{i-1}\Q_1\big(f_{i-1}\,g_1+f_1\,g_{i-1}\big)-a_i\Q_i\Q_1\big(f_i\,g_1+f_1\,g_i\big)\bigg) \quad i \geq 2.
\end{equation}

Neglecting in \eqref{eq:hi-equation} the quadratic term $\Gamma(h,h)$,
one is faced with the linearized problem:
\begin{equation}
  \label{eq:linearized0}
  \dfrac{\d }{\d t} h(t)=L(h(t)),\qquad h(0)=h^{0},
\end{equation}
which should be understood as the linear approximation of equation
\eqref{eq:BD} close to the equilibrium $(\Q_i)_{i \geq 1}$. Our
purpose in the rest of this section is to study the operator $L$ and
its spectral properties, thus obtaining the asymptotic behavior of
equation \eqref{eq:linearized0}.

\subsection{Study of the linearized operator and proof of existence of
  a spectral gap}
\label{sec:linearized}

The first thing we need to do in order to define rigorously the
operator $L$ is to give its domain. We take expression
\eqref{eq:weak-linear} as a starting point, and we denote by
$\ell_{00}$ the set of compactly supported sequences $h = (h_i)_i$
(this is, the sequences for which there exists $N > 0$ such that $h_i
= 0$ for all $i \geq N$).
\begin{defi}[The operator $L$]
  \label{def:L-compact}
  Assume Hypothesis \ref{hyp:linear-section}. For a compactly
  supported sequence $h = (h_i)_i \in \ell_{00}$ we define $L(h)$ by
  the expression \eqref{eq:L} (or, equivalently,
  \eqref{eq:weak-linear} for $h, \phi \in \ell_{00}$).
\end{defi}
Notice that the only infinite sum in \eqref{eq:L} converges by
\eqref{eq:A}. (Actually, the slightly weaker condition $\sum_i a_i
\Q_i < +\infty$ would be enough for the definition, but we keep
\eqref{eq:A} for simplicity).

\medskip
One can give a more compact expression of $L$ by direct
inspection, using \eqref{eq:DB-Qizi} repeatedly:

\begin{lem}
  \label{lem:sigma}
  Assume Hypothesis \ref{hyp:linear-section}. For any compactly
  supported sequence $h$,
  \begin{equation}
    \label{expr}
    L_i(h)=-\sigma_i h_i
    +\sum_{j=1}^\infty \xi_{i,j}h_j
    \quad \text{ for } i \geq 1
  \end{equation}
  where $\sigma_i$ are defined by
  \begin{equation}
    \label{eq:sigma-i}
    \sigma_1 = 3 a_1\Q_1+\sum_{i=1}^\infty a_i \Q_i,
    \qquad
    \sigma_i = a_i\Q_1+b_i
    \quad \text{ for } i \geq 2,
  \end{equation}
  $\xi_{i,j}$ are defined by
  \begin{align}
    \label{eq:xi-12}
    \Q_1\, \xi_{1,2}
    &= \Q_2\, \xi_{2,1}
    = 2 b_2 \Q_2 - a_2\Q_1\Q_2,
    \\
    \label{eq:xi-i1}
    \Q_i\, \xi_{i,1}
    &= \Q_1\, \xi_{1,i}
    = b_i \Q_i - a_i\Q_1\Q_i
    \quad \text{ for } i > 2,
    \\
    \label{eq:xi-ij}
    \Q_i\, \xi_{i,i-1}
    &= \Q_{i-1}\, \xi_{i-1,i}
    = b_{i} \Q_{i}
    \quad \text{ for } i > 2,
  \end{align}
  and $\xi_{i,j}=0$ for $j\notin \{1,i-1,i+1\}.$
\end{lem}

\medskip The numbers $\xi_{i,j}$ represent the nondiagonal entries of
the infinite-dimensional matrix that defines $L$. One can see that the
only nonzero entries of this matrix are in the diagonal and in the
first line and column. In addition, the numbers $\xi_{i,j}$ have the
important property that $(\Q_i \xi_{i,j})_{i,j \geq 1}$ is a symmetric
matrix, which suggests considering the inner product with weight
$\Q_i$. There is another reason why this inner product appears
naturally: since the nonlinear equation (\ref{eq:BD}) has a Lyapunov
functional (see \eqref{eq:def-V}--\eqref{eq:H-thm-BD-Fz}), one may
look at the second-order approximation of this functional close to the
equilibrium $(\Q_i)_{i \geq 1}$, which happens to be
\begin{equation}
  \label{eq:entropy-linear}
  \|h\|_\H^2 := \frac{1}{2} \sum_{i=1}^\infty\Q_i h_i^2.
\end{equation}
It becomes then natural to study the linearized operator $L$
in the Hilbert space $\mathcal{H}:=\ell^2(\Q)$ defined as
\begin{equation*}
  \mathcal{H}
  = \big\{
  h=(h_i)_{i \geq 1} \ \mid\
  \|h\|_{\mathcal{H}} < \infty
  \big\}
\end{equation*}
with inner product denoted by $\langle\cdot,\cdot\rangle$. As remarked
above, $L$ becomes a \emph{symmetric} operator in $\H$. Another
crucial property of $L$ in this Hilbert space is that it is
\emph{dissipative}; this is,
\begin{equation*}
  \langle L h, h \rangle \leq 0
  \fa h \in \ell_{00},
\end{equation*}
as can be readily seen from \eqref{eq:weak-linear}. Though $L$ is in
general not continuous on $\H$, it is easy to give a dense subspace of
$\H$ in which it is bounded:

\begin{lem}
  \label{lem:L-bounded-with-weight}
  Assume Hypothesis \ref{hyp:linear-section} and
  \eqref{eq:Qi-limit}. Then, there exists $C > 0$ (depending only on
  $z$ and the coefficients $(a_i)_{i \geq 1}$, $(b_i)_{i \geq 1}$)
  such that
  \begin{equation}
    \label{eq:L-bounded-with-weight}
    \| L(h) \|_\H
    \leq C \|h\|_{\H_{2}}
  \end{equation}
  for any compactly supported sequence $h$, where
  \begin{equation}
    \label{eq:def-DL}
    \H_{2} = \ell^2(\Q (1+\sigma^2))
    :=
    \Big\{
    h \in \mathcal{H}
    \mid
    \|h\|_{\H_2}^2
    := \sum_{i=1}^\infty \Q_i  (1+\sigma_i^2) h_i^2 < +\infty
    \Big\}
  \end{equation}
  and $(\sigma_i)_i$ was defined in \eqref{eq:sigma-i}.
\end{lem}

\begin{rem}
  In fact, instead of condition \eqref{eq:Qi-limit} it is enough to
  have that $Q_{i+1}/Q_i$ is uniformly bounded in $i$, as can be seen
  from the proof. Notice that the assumptions (H1)--(H4) of
  \cite{JN03} imply the boundedness of $(Q_{i+1}/Q_i)_i$.
\end{rem}

\begin{proof}
  Instead of \eqref{eq:L-bounded-with-weight} we will show,
  equivalently, that
  \begin{equation}
    \label{eq:L-bounded-weak}
    \langle L h, \phi \rangle \leq C \|h\|_{\H_2} \|\phi\|_\H
    \fa \phi \in \H.
  \end{equation}
  In this proof $C$ denotes any positive quantity depending only on
  $z$ and the coefficients $(a_i)_{i \geq 1}$, $(b_i)_{i \geq 1}$,
  possibly changing from line to line. From \eqref{eq:weak-linear} we
  have, using Cauchy-Schwarz's inequality,
  \begin{multline}
    \label{eq:pr1}
    \langle L h, \phi \rangle
    =
    \sum_{i=1}^\infty a_{i}\Q_i \Q_1 \big( h_i +
    h_1 - h_{i+1} \big) \big(\phi_{i+1}-\phi_i-\phi_1\big)
    \\
    \leq
    \Q_1 \left( \sum_{i=1}^\infty a_{i}^2 \Q_i
      \big( h_i +  h_1 - h_{i+1} \big)^2 \right)^{1/2}
    \left(  \sum_{i=1}^\infty \Q_i
      \big(\phi_{i+1}-\phi_i-\phi_1\big)^2 \right)^{1/2}.
  \end{multline}
  The term inside the first parentheses is bounded by
  \begin{multline*}
     \sum_{i=1}^\infty a_{i}^2 \Q_i
     \big( h_i +  h_1 - h_{i+1} \big)^2
     \leq
     3 \sum_{i=1}^\infty a_{i}^2 \Q_i
     \big( h_i^2 +  h_1^2 + h_{i+1}^2 \big)
     \\
     \leq
     3 \sum_{i=1}^\infty \sigma_i^2 \Q_i h_i^2
     +
     3 A h_1^2
     +
     3 \sum_{i=1}^\infty a_i^2 \Q_i h_{i+1}^2
     \\
     \leq C \|h\|_{\H_2}^2
     + \sum_{i=1}^\infty  \frac{\Q_{i+1}}{\Q_{i}}
     b_{i+1}^2 \Q_{i+1} h_{i+1}^2
     \leq C \|h\|_{\H_2}^2,
   \end{multline*}
   where we used \eqref{eq:DB-Qizi} and the fact that $\Q_{i+1} /
   \Q_{i}$ is bounded uniformly in $i$ thanks to
   \eqref{eq:Qi-limit}. A similar reasoning shows that the second
   parenthesis in \eqref{eq:pr1} is bounded by
   \begin{equation*}
      \sum_{i=1}^\infty \Q_i
      \big(\phi_{i+1}-\phi_i-\phi_1\big)^2
      \leq
      C \|\phi\|_{\H}^2,
    \end{equation*}
    which shows \eqref{eq:L-bounded-weak}.
\end{proof}

\medskip

The operator $L$ is defined only on compactly supported sequences
$h$. We will now extend it to a larger domain, so that it becomes a
closed operator, and study its spectrum. We will consider two ways to
do this, later shown to lead to the same result: one of them is to
look at the closure $\mathbf{L}$ of $L$, and the other one is to look
at an associated quadratic form, which will lead to an extension
denoted by $\L$. Let us first describe the latter.

Expression \eqref{eq:weak-linear} suggests the introduction of a
symmetric form $\mathscr{E}$ in $\H$ by setting
\begin{equation*}
  \mathscr{E}(h,g) = -\langle h, L g \rangle
  = \sum_{i\geq 1}
  a_{i}\Q_i \Q_1
  \big(h_{i+1} - h_i - h_1 \big)
  \big(g_{i+1} - g_i - g_1 \big)
  \quad \text{ for } h,g \in \ell_{00},
\end{equation*}
which can be naturally extended to the domain
\begin{equation*}
  \mathscr{D}(\mathscr{E}) =
  \left\{h \in \mathcal{H}\;;\;\mathscr{E}(h,h) < \infty\right\}
\end{equation*}
by setting
\begin{equation}
  \label{eq:def-quadratic}
  \mathscr{E}(h,g) =
  \sum_{i\geq 1} a_{i}\Q_i \Q_1
  \big(h_{i+1} - h_i - h_1 \big)
  \big(g_{i+1} - g_i - g_1 \big)
  \quad \text{ for } h,g \in \mathscr{D}(\mathscr{E}).
\end{equation}
Then, one has the following
\begin{prp}
  \label{prp:quadratic-form}
  Assume Hypothesis \ref{hyp:linear-section}. Define the space $\H_1$
  by
  \begin{equation*}
    \H_1 =
    \bigg\{
    h \in \mathcal{H}\,;\,
    \|h\|_{\mathcal{H}_1}^2 := \sum_{i=1}^\infty (1+\sigma_i) \Q_i
    h_i^2 < \infty\,
    \bigg\}
    ,
  \end{equation*}
  where we recall that $\sigma_i$ was defined in \eqref{eq:sigma-i}.

  The form $\mathscr{E}$ with domain $\mathscr{D}(\mathscr{E})$
  defined by \eqref{eq:def-quadratic} is a closed symmetric form on
  $\mathcal{H}$. Its domain $\mathscr{D}(\mathscr{E})$ contains $\H_1$ and
  there exists $C >0$ (depending only on the coefficients
  $(a_i)_{i\geq 1}$ and $(b_i)_{i \geq 2}$) such that
  \begin{equation}
    \label{DL}
    |\mathscr{E}(h,g)|
    \leq
    C \|h\|_{\mathcal{H}_1}\|g\|_{\mathcal{H}_1}
    \fa h, g \in \mathcal{H}_1.
  \end{equation}
\end{prp}

\begin{proof}
  The fact that $\mathscr{E}$ is symmetric and bilinear is clear.
  Moreover, one proves as in \cite[Example 1.2.4]{fukushima} that
  $(\mathscr{E},\mathscr{D}(\mathscr{E}))$ is closed thanks to Fatou's
  Lemma.  Let us prove \eqref{DL}, which in turn shows that $\H_1
  \subseteq \mathscr{D}(\mathscr{E})$. By a usual depolarization
  argument it is enough to show \eqref{DL} for $g=h$. We have
  \begin{multline*}
    \mathscr{E}(h,h)
    =
    \sum_{i=1}^\infty a_{i }\Q_i \Q_1
    \big(h_i+h_1 -h_{i+1} \big)^2
    \leq
    3 \sum_{i=1}^\infty a_{i }\Q_i \Q_1
    \big(h_i ^2 + h_1 ^2 + h_{i+1} ^2 \big)
    \\
    =3 A h_1^2 +
    3
    \sum_{i=1}^{\infty} a_i \Q_1 \Q_i h_i ^2
    +
    3 \sum_{i=2}^\infty b_i \Q_i h_i^2
    \leq C \|h\|_{\mathcal{H}_1}^2,
  \end{multline*}
  where $A$ is the quantity in \eqref{eq:A}. This proves the desired
  bound.
\end{proof}

\bigskip

Due to the previous Proposition, according to \cite[Theorem 1.3.1 \&
Corollary 1.3.1]{fukushima}, there exists a unique non-positive
definite self-adjoint operator $\mathcal{L}$ on $\mathcal{H}$ such
that $\mathscr{D}(\L) \subset \mathscr{D}(\mathscr{E})$ and
\begin{equation}
  \label{eq:def-EL}
  \mathscr{E}(h,g) = -\langle \L h, g \rangle
  \fa  h \in \mathscr{D}(\L), \ g \in \mathscr{D}(\mathscr{E}).
\end{equation}
More precisely, $\mathscr{D}(\mathscr{E}) = \mathscr{D}(\sqrt{-\L}).$
It is clear that the linear operator $(\L,\D(\L))$ extends the above
linear operator $L$, defined on $\ell_{00}$, and that $\ell_{00}$ is a
core for $\L$. It is also easy to see, under the conditions of Lemma
\ref{lem:L-bounded-with-weight}, that the domain of $\L$ must include
the space $\H_2$ (since $\L$ is a closed operator that extends $L$),
and that the expression of $\L$ in $\H_2$ is still given by
\eqref{eq:L} (since each of the sums converges absolutely in this
space, as deduced from the proof of Lemma
\ref{lem:L-bounded-with-weight}) or alternatively by
\eqref{eq:weak-linear} (for any $\phi \in \ell_{00}$). It is also easy
to see that $0$ is an eigenvalue of $\L$ with explicit eigen-space:

\begin{lem}
  \label{lem:0-eigenspace}
  Assume Hypothesis \ref{hyp:linear-section} and
  \eqref{eq:Qi-limit}. Then $0$ is an eigenvalue of $\L$, with a
  one-dimensional associated eigenspace spanned by the sequence
  defined by $h_i = i$ for $i \geq 1$.
\end{lem}

\begin{proof}
  Under the condition \eqref{eq:A} one sees that $h = (i)_{i \geq 1}
  \in \H_2 \subseteq \mathscr{D}(\L)$ (notice that the latter
  inclusion holds due to the previous discussion and Lemma
  \ref{lem:L-bounded-with-weight}). It is clear then from
  \eqref{eq:def-EL} and \eqref{eq:def-quadratic} that $\ap{\L h, g} =
  0$ for all $g \in \mathscr{D}(\mathscr{E})$, so $\L h = 0$. On the
  other hand, if there is any $\tilde{h} = (\tilde{h}_i)_{i \geq 1}
  \in \mathscr{D}(\L)$ such that $\L(\tilde{h})=0$, then
  $\mathscr{E}(h,h) = 0$ and consequently $h_{i+1} = h_i + h_1$ for
  all $i \geq 1$. This implies that $\tilde{h}_i = i h_1$ for $i \geq
  1$, so $\tilde{h}$ is a multiple of $h$.
\end{proof}

\begin{rem}
  Again, Lemma \ref{lem:0-eigenspace} holds with the milder condition
  that $Q_{i+1}/Q_i$ is uniformly bounded instead of
  \eqref{eq:Qi-limit}.
\end{rem}
It is well-known \cite{kato,fukushima} that the Dirichlet form
$(\mathscr{E},\mathscr{D}(\mathscr{E}))$ (or, equivalently, the
operator $(\L,\D(\L))$), generates a strongly continuous semigroup of
contractions $\{S_t,\;t \geq 0\}$ in $\mathcal{H}=\ell^2(\Q)$. In
particular, for any $h^0 = (h_i^0)_{i \geq 1} \in \mathcal{H}$ the
linearized problem
\begin{equation}
  \label{eq:linearized}
  \dfrac{\d }{\d t} h_i(t)=\L_i(h(t)),\qquad h_i(0)=h_i^{0},
\end{equation}
admits a unique mild solution $h(t)=(h_i(t))_{i \geq 1}$ given by
$h(t)= S_th^0$ satisfying
$$\|h(t)\| \leq \|h^0\| \fa t \geq 0.$$
Moreover, $S_t(\mathcal{H}) \subset \mathscr{D}(\mathscr{E})$ for any
$t > 0$  and, if $h^0=(h_i^0)_{i \geq 1}
\in \D(\L)$ then, for any $t \geq 0$, $\sup_{s\in
  [0,t)}\|h(s)\|_{\mathscr{D}(\mathscr{E})} < \infty$ (we refer to
\cite[Chapter 1]{fukushima} for details). In fact, for $h^0 \in
\D(\L)$ we have that $h$ is $\mathcal{C}^1$ in time, $h(t) \in
\mathscr{D}(\L)$ for all $t \geq 0$ and \eqref{eq:linearized} is
satisfied pointwise in time. Hence, for $h^0 \in
\D(\L)$ we have
\begin{equation*}
  \frac{\d}{\d t} \|h(t)\|^2 = \langle h, \L h \rangle =
  -\mathscr{E}(h,h)
  =
  -\sum_{i\geq 1} a_{i}\Q_i \Q_1
  \big(h_{i+1} - h_i - h_1 \big)^2
  \fa t \geq 0.
\end{equation*}

The advantage of the approach involving the quadratic form is that we
obtain, in a simple way, a self-adjoint extension of $L$ on which we
have a lot of information. However, it is not easy to give explicitly
the domain $\mathscr{D}(\L)$ of the operator $\L$. We consider also a
different approach based on the closure of $L$. Since $L$ is
symmetric, it is closable and we define $\mathbf{L}$ as its closure:

\begin{defi}[The operator $\mathbf{L}$]
  \label{def:L}
  Assume Hypothesis \ref{hyp:linear-section}. We define the linear
  operator $\mathbf{L}$, with domain $\mathscr{D}(\mathbf{L})$, as the
  closure in $\mathcal{H}$ of the linear operator $L$ (with domain
  $\ell_{00}$).
\end{defi}

The following proposition gives a certain structure to $\mathbf{L}$
that will allow us to identify its domain, prove it has a spectral gap
(in a non-constructive way, using Weyl's theorem) and show that in
fact $\mathbf{L} = \L$:

\begin{prp}
  \label{prp:L-decomposition}
  Assume Hypothesis \ref{hyp:linear-section} and conditions
  \eqref{eq:ai-bounded-below}, \eqref{eq:Qi-limit} and
  \eqref{eq:ai-limit}. Assume also that $z < \zs$. Then the linear
  operator $\mathbf{L}$ is \emph{self-adjoint} with domain
  $\mathscr{D}(\mathbf{L})=\mathcal{H}_2$ (as defined in
  \eqref{eq:def-DL}) and, given $\delta > 0$, it can be written as
  \begin{equation}
    \label{eq:L-decomposition}
    \mathbf{L}(h)
    = \mathbf{L}^{\mathrm{C}}(h) + \mathbf{L}^{\mathrm{M}}(h)
    \qquad \text{ for all } h \in \mathscr{D}(\mathbf{L})
  \end{equation}
  such that
  \begin{enumerate}
  \item $\mathbf{L}^{\mathrm{C}}$ is a compact operator on $\mathcal{H}$.
  \item $\mathbf{L}^{\mathrm{M}}$ is a self-adjoint operator with domain
    $\mathcal{H}_2$, and with spectrum
    \begin{equation*}
      \mathfrak{S}(\mathbf{L}^{\mathrm{M}})
      \subseteq (-\infty, -\lambda_{\mathrm{M}}]
    \end{equation*}
    where
    \begin{equation*}
      \lambda_\mathrm{M} :=
      \underline{\sigma}\,\left(1-2\frac{\sqrt{\ell}}{1+\ell}\right)
      - \delta,
      \qquad \ell := \zs / z
    \end{equation*}
    and
    \begin{equation*}
      \underline{\sigma} := \inf_{k \geq 1} \sigma_k,
    \end{equation*}
    which is strictly positive due to \eqref{eq:ai-bounded-below}.
  \end{enumerate}
\end{prp}

\begin{proof}
  Take an integer $N \geq 1$, to be fixed later. We define, for all $i
  \geq 1$,
  \begin{gather}
    \label{eq:def-LC-0}
    \mathbf{L}^{\mathrm{C}}_i(h)
    :=
    \sum_{j=1}^\infty \chi_{\{\min\{i,j\} \leq N\}} \, \xi_{i,j} h_j
    \quad \text{ for } h \in \H,
    \\
    \label{eq:def-LM-0}
    \mathbf{L}^{\mathrm{M}}_i(h)
    :=
    - \sigma_i h_i
    + \sum_{j=1}^\infty \chi_{\{\min\{i,j\} > N\}} \, \xi_{i,j} h_j
    \quad \text{ for } h \in \H_2,
  \end{gather}
  We recall that the notation $\sigma_i$ and $\xi_{i,j}$ was defined
  in \eqref{eq:sigma-i}--\eqref{eq:xi-ij}. The notation
  $\chi_{\{\cdots\}}$ represents a function which is equal to $1$ when the
  condition in the brackets is satisfied, $0$ otherwise. In a more
  explicit way,
  \begin{equation}
    \begin{cases}
      \label{eq:def-LC}
      \mathbf{L}^{\mathrm{C}}_1(h)
      = \sum_{j=2}^\infty\xi_{1,j}\,h_j,
      \qquad
      \mathbf{L}^{\mathrm{C}}_2(h)
      = \xi_{2,1}\,h_1 + a_2 \Q_1 \,h_3
      \\
      \mathbf{L}^{\mathrm{C}}_i(h) =
      \xi_{i,1}h_1 + b_i\,h_{i-1} \,\chi_{\{i \leq N+1\}}
      + a_i\Q_1\,h_{i+1} \,\chi_{\{i\leq N\}}
      \quad \text{ for } i > 2;
    \end{cases}
  \end{equation}
  while
  \begin{equation}
    \label{eq:def-LM}
    \mathbf{L}^{\mathrm{M}}_i(h)
    = -\sigma_i\,h_i+b_i\,h_{i-1} \,\chi_{\{i > N+1\}}
    + a_i\,\Q_1\,h_{i+1} \,\chi_{\{i >N\}}
    \qquad (i \geq 1).
  \end{equation}
  With calculations similar to those in the proof of Lemma
  \ref{lem:L-bounded-with-weight} one sees that the sums in
  \eqref{eq:def-LC-0}--\eqref{eq:def-LM-0} converge, that
  $\mathbf{L}^{\mathrm{C}} : \mathcal{H}\to \mathcal{H}$ is bounded and
  that there exists $C >0$ such that
  $\|\mathbf{L}^{\mathrm{M}}\,h\|_{\mathcal{H}} \leq C\|h\|_{\mathcal{H}_2}$ for any
  $h \in \mathcal{H}_2$. Moreover, both $\mathbf{L}^\mathrm{C}$ and
  $\mathbf{L}^{\mathrm{M}}$ are easily seen to be symmetric operators.
  One sees that $\mathbf{L}^{\mathrm{C}} : \mathcal{H}\to \mathcal{H}$
  is a finite-rank operator, since except for its first component
  $\mathbf{L}^{\mathrm{C}}_1$ it only depends on a finite number of
  components of $h$. In particular, it is a compact operator. Since it
  is also symmetric, $\mathbf{L}^\mathrm{C}$ is thus self-adjoint.


  Let us investigate now the remaining part
  $\mathbf{L}^{\mathrm{M}}$. We will prove that, for $N >1$ large
  enough, $\mathbf{L}^{\mathrm{M}}$ is self-adjoint thanks to
  Kato-Rellich's theorem \cite[Theorem 4.3, p. 287]{kato}. To do so,
  define the multiplication operator $\mathbf{T}$ with domain $\H_2$ by
  \begin{equation*}
    \mathbf{T}_i(h) = -\sigma_i\,h_i
    \quad \text{ for } h \in \H_2,\ i \geq 1.
  \end{equation*}
  Being a multiplication operator, $\mathbf{T}$ is self-adjoint in
  $\mathcal{H}$. Define the operator $\mathbf{S}$, also with domain
  $\H_2$, as the remaining part in \eqref{eq:def-LM}:
  \begin{equation*}
    \mathbf{S}_i(h)
    = b_i\,h_{i-1}\chi_{\{i > N+1\}}
    + a_i\,\Q_1\,h_{i+1}\chi_{\{i >N\}}
    \quad \text{ for } h \in \H_2,\ i \geq 1,
  \end{equation*}
  so that $\mathbf{L}^M = \mathbf{T} + \mathbf{S}$. Clearly,
  $\mathbf{S}$ is symmetric. Let us now prove that $\mathbf{S}$ is
  $\mathbf{T}$-bounded with relative bound smaller than one. To do so,
  we compute
  \begin{equation*}
    \|\mathbf{T}(h)\|_\H =
    \left( \sum_{i=1}^\infty \sigma_i^2 \Q_i h_i^2 \right)^{1/2},
  \end{equation*}
  while
  \begin{multline}
    \label{Sh}
    \|\mathbf{S}(h)\|_\H =
    \left( \sum_{i=1}^\infty |\mathbf{S}_i(h)|^2 \Q_i \right)^{1/2}
    \leq
    \left(
      \sum_{i > N+1} b_i^2 h_{i-1}^2\Q_i
    \right)^{1/2}
    + \left(
      \sum_{i > N} (a_i\Q_1)^2 h_{i+1}^2\,\Q_i
    \right)^{1/2}
    \\
    =
    \left(
      \sum_{i > N}
      b_{i+1}^2 \Q_{i+1} h_i^2
    \right)^{1/2}
    + \left(
      \sum_{i > N+1}
      a_{i-1}^2 \Q_1^2 \Q_{i-1}
      h_i^2 \right)^{1/2}
    \\
    =
    \left(
      \sum_{i > N}
      \mu_i
      \Q_i h_i^2 \right)^{1/2}
    + \left(
      \sum_{i > N+1}
      \nu_i
      \Q_i h_i^2 \right)^{1/2}
  \end{multline}
  where
  \begin{equation*}
    \nu_i := a_{i-1}^2 \Q_1^2 \frac{\Q_{i-1}}{\Q_i},
    \qquad
    \nu_{i-1}
    = a_i^2 \Q_1^2 \frac{\Q_i}{\Q_{i+1}}
    = b_{i+1}^2 \frac{\Q_{i+1}}{\Q_i}.
  \end{equation*}
  Observe that, since $\sigma_i = a_i \Q_1 + b_i = \Q_1 \big(a_i +
  a_{i-1} \frac{\Q_{i-1}}{\Q_i}\big)$, and calling $\ell
  := \zs / z$,
  \begin{equation*}
    \frac{\nu_i}{\sigma_i^2}
    =
    \frac{a_{i-1}^2}{\left(a_i + a_{i-1} \frac{\Q_{i-1}}{\Q_i}\right)^2}
    \frac{\Q_{i-1}}{\Q_i}
    \longrightarrow \frac{\ell}{(1 + \ell)^2}
    \quad \text{ as } i \to +\infty,
  \end{equation*}
  due to \eqref{eq:Qi}, \eqref{eq:Qi-limit} (which implies $\Q_{i} /
  \Q_{i+1} \to \ell$ as $i \to +\infty$) and
  \eqref{eq:ai-limit}. Similarly,
    \begin{equation*}
    \frac{\nu_{i-1}}{\sigma_i^2}
    \longrightarrow \frac{\ell}{(1 + \ell)^2}
    \quad \text{ as } i \to +\infty.
  \end{equation*}
  Hence, from \eqref{Sh} we see that for any $\epsilon> 0$ we can find
  $N > 0$ such that
  \begin{equation}
    \label{eq:Sth}
    \|\mathbf{S}(h)\|_\H
    \leq
    (1+\epsilon) \frac{2 \sqrt{\ell}}{1+\ell}
    \left(
      \sum_{i > N}
      \sigma_i^2 \Q_i h_i^2
    \right)^{1/2}
    \leq
    (1+\epsilon) \frac{2 \sqrt{\ell}}{1+\ell}
    \,
    \|\mathbf{T}(h)\|_\H
    =: \theta \|\mathbf{T}(h)\|_\H.
  \end{equation}
  Since we are assuming $z < \zs$, we have $\ell > 1$ and
  it is possible to choose $\epsilon > 0$ such that $\theta = 2
  (1+\epsilon) \sqrt{\ell} / (1+\ell) < 1$.  In other words,
  $\mathbf{S}$ is $\mathbf{T}$-bounded with relative bound $\theta$ strictly
  smaller than $1$. According to Kato-Rellich's Theorem
  $\mathbf{L}^{\mathrm{M}} = \mathbf{S} + \mathbf{T}$ with domain
  $\D(\mathbf{L}^{\mathrm{M}}) = \D(\mathbf{T}) = \mathcal{H}_2$ is
  self-adjoint.

  Let us now show that the spectrum of $\mathbf{L}^{\mathrm{M}}$ is to
  the left of $\lambda_\mathrm{M}$. Since $\underline{\sigma} =
  \inf_{i \geq 1}\sigma_i > 0$, one has $\langle
  \mathbf{T}(h),h\rangle \leq -\underline{\sigma} \|h\|^2$ for any $h
  \in \D(\mathbf{T})$. With the terminology of \cite{kato}, this
  exactly means that the multiplication operator $\mathbf{T}$ is
  bounded from above with upper bound $\gamma_\mathbf{T} =
  -\underline{\sigma}$. Combining inequality \eqref{eq:Sth} with
  \cite[Theorem 4.11, p. 291]{kato}, $\mathbf{L}^{\mathrm{M}}$ is
  bounded from above with upper bound
  $\gamma_{\mathbf{L}^{\mathrm{M}}}$ such that
  $\gamma_{\mathbf{L}^{\mathrm{M}}} \geq \gamma_\mathbf{T} -
  \theta|\gamma_\mathbf{T}| = -(1-\theta) \underline{\sigma}$, i.e.
  \begin{equation*}
    \langle \mathbf{L}^{\mathrm{M}} h,\,h \rangle
    \leq -(1-\theta) \underline{\sigma} \|h\|^2
    \fa h \in \D(\mathbf{T}).
  \end{equation*}
  This proves in particular point \textit{(2)} of the Proposition with
  $\lambda_\mathrm{M} = 1-\theta > 0$. Notice that we can take $\epsilon
  > 0$ arbitrarily small, so $\lambda_M$ can be arbitrarily close to
  $2\sqrt{\ell}/(1+\ell)$.

  Now the identity $L(h) = \mathbf{L}^{\mathrm{C}}(h) +
  \mathbf{L}^{\mathrm{M}}(h)$ is obviously true for $h \in
  \ell_{00}$. Then, it is clear from the bounds above that the closure
  of $L$ has domain $\H_2$ and is equal to $\mathbf{L}^{\mathrm{C}} +
  \mathbf{L}^{\mathrm{M}}$, which shows
  \eqref{eq:L-decomposition}. Since $\mathbf{L}$ is a compact
  self-adjoint perturbation of $\mathbf{L}^{\mathrm{M}}$, Weyl's
  Theorem ensures that $\mathbf{L}$ is self-adjoint.
\end{proof}

\begin{cor}
  \label{cor:non-constructive-gap}
  Assume the conditions of Proposition
  \ref{prp:L-decomposition}. There exists $\lambda_0 >0$ (defined in a
  non-constructive way) such that
  \begin{equation}
    \label{eq: spcgap}
    \langle\, \mathbf{L} h, h \rangle \leq -\lambda_0 \|h\|^2
    \qquad
    \fa h \in \mathcal{H}_2 \text{ such that }
    \sum_{i=1}^\infty i h_i \Q_i = 0.
  \end{equation}
  In addition, the operators $\mathbf{L}$ and $\L$ are equal. Finally, $\mathbf{L}$ is the generator of a $C_0$-semigroup $(\mathcal{S}_t)_{t \geq 0}$ in $\mathcal{H}$ such that
  $$\|\mathcal{S}_t h\|_\mathcal{H} \leq e^{-\lambda_0 t} \|h\|_\mathcal{H} \qquad \forall h \in \mathcal{H} \text{ such that }
    \sum_{i=1}^\infty i h_i \Q_i = 0, \qquad t \geq 0.$$
\end{cor}

\begin{proof}
  Since $\mathbf{L}^{\mathrm{C}}$ is a compact perturbation of
  $\mathbf{L}^{\mathrm{M}}$, Weyl's Theorem ensures that the operator
  $\mathbf{L}$ is self-adjoint and its essential spectrum
  $\mathfrak{S}_{\mathrm{ess}}(\mathbf{L})$ coincides with that of
  $\mathbf{L}^{\mathrm{M}}$, in particular
  $\mathfrak{S}_{\mathrm{ess}}(\mathbf{L}) \subseteq (-\infty,
  -\lambda_\mathrm{M}]$. Consequently, the part of
  $\mathfrak{S}(\mathbf{L})$ contained in $(-\lambda_{\mathrm{M}},
  +\infty)$ is a set of isolated eigenvalues of finite
  multiplicity. As we also know that $\mathbf{L}$ is non-positive, we
  have $\mathfrak{S}(\mathbf{L}) \subseteq (-\infty, 0]$, having only
  a finite number of eigenvalues in $(-\lambda_{\mathrm{M}}, 0]$. From
  Lemma \ref{lem:0-eigenspace} we know that $0$ is an eigenvalue of
  $\mathbf{L}$ with eigenspace spanned by $(i)_{i \geq 1}$. Since
  $\lambda_\mathrm{M} > 0$, $\mathbf{L}$ must have a strictly positive
  spectral gap $\lambda_0$, which gives \eqref{eq: spcgap}.

  Since $\mathbf{L}$, the closure of $L$, is self-adjoint, the
  operator $L$ is by definition essentially self-adjoint. As such, it
  only has one possible self-adjoint extension. Since $\L$ is another
  such extension, we have $\mathbf{L} = \L$. The final part of the result is a classical consequence of \eqref{eq: spcgap}.
\end{proof}

Proposition \ref{prp:L-decomposition} does not apply in the case $z =
\zs$. In order to include that case we need to obtain more
delicate estimates:
\begin{lem}
  \label{lem:deltak}
  Assume the hypotheses of Proposition \ref{prp:L-decomposition}, but
  take $z = \zs$. Define then
  \begin{equation*}
    \delta_k = \frac{b_k}{a_k \Q_1}-1
    = \frac{a_{k-1}}{ a_k} \frac {\Q_{k-1}}{\Q_k}-1
    \quad \text{ for } k \geq 1.
  \end{equation*}
  and assume that
  \begin{equation}
    \label{eq:deltak2}
    \liminf_{k \to \infty} \delta_k \, \sqrt{a_k} > 0.
  \end{equation}
  Then the conclusions of Proposition \ref{prp:L-decomposition} still
  hold true for $\mathbf{L}$, for some $\lambda_\mathrm{M} > 0$. As a
  consequence, the conclusions of Corollary
  \ref{cor:non-constructive-gap} are true (i.e., $\mathbf{L}$ has a
  positive spectral gap and $\mathbf{L} = \L$).
\end{lem}

\begin{proof}
  The proof follows the same steps as that of Proposition
  \ref{prp:L-decomposition} and is based upon the splitting of
  $\mathbf{L}$ as
  $\mathbf{L}=\mathbf{L}^{\mathrm{C}}+\mathbf{L}^{\mathrm{M}}$ with
  $\mathbf{L}^{\mathrm{C}}$ and $\mathbf{L}^{\mathrm{M}}$ defined by
  \eqref{eq:def-LC} and \eqref{eq:def-LM} respectively. It will
  consist in proving that, under the supplementary assumption
  \eqref{eq:deltak2}, one can choose $N >1$ large enough such that
  \begin{equation}\label{eq:dissipative}
    \ap{\mathbf{L}^{\mathrm{M}}(h), h}
    \leq
    -\lambda_\mathrm{M} \|h\|_{\mathcal{H}}^2
    \fa h \in \mathcal{H}_2
  \end{equation}
  holds true for some positive $\lambda_\mathrm{M} >0$. Once we have
  this, the rest of the proofs of Proposition
  \ref{prp:L-decomposition} and Corollary
  \ref{cor:non-constructive-gap} are still valid. Indeed, arguing as
  in the proof of Proposition \ref{prp:L-decomposition}, this implies
  the existence of a positive spectral gap (defined in a
  non-constructive way) for $\mathbf{L}$ thanks to Weyl's theorem.
  One computes, for some $N >1$ to be fixed later,
  \begin{multline}
    \label{eq:LMhh}
    \ap{\mathbf{L}^{\mathrm{M}}(h), h}
    = -\sum_{i=1}^\infty \sigma_i \Q_i \,h_i^2
    + \sum_{i > N+1}b_i\,\Q_i h_i\,h_{i-1}
    + \sum_{i > N} a_i\,\Q_1 \Q_i h_i\,h_{i+1}\\
    = -\sum_{i=1}^\infty \sigma_i \Q_i \,h_i^2
    + 2 \sum_{i > N} a_i\,\Q_1 \Q_i h_i\,h_{i+1}
  \end{multline}
  where we used \eqref{eq:Qi}. For each $i \geq 1$ we take $r_i > 1$
  (also to be fixed later), and use Young's inequality to deduce that
  \begin{multline*}
    2 \sum_{i > N} a_i\,\Q_1 \Q_i h_i\,h_{i+1}
    \leq
    \sum_{i > N+1} r_i a_i \Q_1 \,h_i^2
    + \sum_{i > N+1} \frac{1}{r_i}
    a_i \Q_1 \Q_i \,h_{i+1}^2
    \\
    =
    \sum_{i > N+1} r_i a_i\,\Q_1 h_i^2 \Q_i
    +
    \sum_{i >N+2}  \frac{1}{r_i} b_i h_i^2 \Q_j,
  \end{multline*}
  again by \eqref{eq:Qi}. Since $b_i = (a_i\Q_1)(1+\delta_i)$ and
  $\sigma_i = b_i+a_i\Q_1=(a_i \Q_1)(2+\delta_i)$ for any $i > 1$, one
  gets
  \begin{multline*}
    \langle \mathbf{L}^{\mathrm{M}}h,\,h\rangle
    \leq
    -\sum_{i=1}^\infty \sigma_i \Q_i h_i^2
    + \sum_{i > N+1} a_i\,\Q_1
    \left(r_i + \dfrac{1+\delta_i}{r_i}\right)h_i^2 \,\Q_i
    \\
    \leq
    -\sum_{i=1}^{N+1} \sigma_i \Q_i\,h_i^2
    + \sum_{i > N+1} a_i \Q_1
    \left(r_i + \frac{1+\delta_i}{r_i}-2-\delta_i\right)
    \Q_i\,h_i^2.
  \end{multline*}
  Take $N$ large enough so that for $\delta_i > 0$ for all $i > N+1$.
  We make the choice
  \begin{equation*}
    r_i = 1 + \delta_i/2
    \quad \text{ for } i > N+1,
  \end{equation*}
  to obtain
  \begin{equation*}
    \langle \mathbf{L}^{\mathrm{M}}h,\,h\rangle
    \leq
    -\sum_{i=1}^{N+1} \sigma_i \Q_i\,h_i^2
    - \Q_1 \sum_{i > N+1}
    \frac{a_i \delta_{i}^2}{4+2\delta_{i}}
    \Q_i\,h_i^2.
  \end{equation*}
  According to \eqref{eq:deltak2},
  \begin{equation*}
    C := \Q_1 \liminf_i \frac{a_i \delta_{i}^2}{4+2\delta_{i}} > 0,
  \end{equation*}
  so choosing $N$ large enough,
  \begin{equation*}
    \langle \mathbf{L}^{\mathrm{M}}h,\,h\rangle
    \leq
    -\sum_{i=1}^{N+1} \sigma_i \Q_i\,h_i^2
    - \frac{C}{2} \sum_{i > N+1} \Q_i\,h_i^2
    \leq
    - \underline{\sigma} \sum_{i=1}^{N+1} \Q_i h_i^2
    - \frac{C}{2} \sum_{i >N+1} \Q_i h_i^2,
  \end{equation*}
  which yields \eqref{eq:dissipative} with
  $\lambda_\mathrm{M}=\min\{\underline{\sigma},\frac{C}{2}\} >0$.
\end{proof}

\begin{rem}
  For the typical coefficients given by \eqref{eq:PT-coefficients},
  with $z=z_s$, one sees that
  $$\delta_k= \dfrac{q}{z_s \,k^{1-\mu}} \qquad k \geq 1.$$
  Since $a_k=k^\alpha$ with $\alpha >0$, one sees that
  \eqref{eq:deltak2} holds true if and only if $\alpha \geq 2(1-\mu).$
  In other words, $\mathbf{L}$ has a positive spectral gap whenever
  $\alpha \geq 2(1-\mu)$. We will prove this later on in a different
  way, by a constructive argument related to Hardy's inequality which
  gives explicit estimates of the size of the gap. We will also show
  that, for $\alpha < 2(1-\mu)$, $\L$ does not have a positive
  spectral gap (see Lemma \ref{lem:B-bounds}).
\end{rem}


\subsection{Explicit spectral gap estimates in $\ell^2(\Q)$.}
\label{sec:explicit-gap}

Let us now study conditions ensuring that $\L$ admits a positive
spectral gap in the space $\mathcal{H}=\ell^2(\Q)$. Since $\L$ is a
self-adjoint operator in this space, it having a spectral gap of size
$\lambda_0$ is equivalent to the functional dissipativity inequality
\begin{equation}
  \label{eq:BDgap}
  \ap{h, \L h}
  \leq -\lambda_0 \|h\|_{\H}^2
  \quad \text{ for all $h \in \ell_{00}$
    such that $\sum_{i=1}^\infty i\,\Q_i h_i = 0$}
\end{equation}
this is,
\begin{equation}
  \label{eq:BDgap2}
  \lambda_0
  \sum_{i=1}^\infty \Q_i h_i^2
  \leq
  \sum_{i=1}^\infty
  a_{i}\Q_i \Q_1
  \big(h_{i+1} - h_i - h_1 \big)^2
\end{equation}
for all $h \in \ell_{00}$ such that $\sum_{i=1}^\infty i\,\Q_i h_i =
0$. Notice that it is enough to have the inequality for compactly
supported sequences $h$, since $\ell_{00}$ is a core for $\L$. On the
other hand, if \eqref{eq:BDgap2} holds for all $h \in \ell_{00}$ then
it must actually hold for \emph{all} sequences $h$, with the
understanding that either or both sides of the inequality may be
infinite.

We already saw in Proposition \ref{prp:L-decomposition} (see also
Lemma \ref{lem:deltak}) sufficient conditions ensuring the existence
of a positive spectral gap. However, this existence result was based
on a compactness argument and, consequently, does not provide any
quantitative information about the size spectral gap. In this section
we adopt a different viewpoint by studying the inequality
\eqref{eq:BDgap} directly. This will result in useful estimates on the
spectral gap, explicit in terms of the coefficients $(a_k)_{k \geq 1}$
and $(b_k)_{k \geq 1}$.

\begin{defi}
  \label{def:BD-spectral-gap}
  Assume Hypothesis \ref{hyp:linear-section}. We call $\lambda_0$ the
  size of the spectral gap in $\mathcal{H}$ of the operator $\L$;
  equivalently, $\lambda_0$ is the largest nonnegative constant such
  that \eqref{eq:BDgap} (or equivalently, \eqref{eq:BDgap2})
  holds. Notice that $\lambda_0 = 0$ whenever $\L$ has no spectral
  gap.
\end{defi}
Our main estimate of $\lambda_0$ is the following:

\begin{thm}[\textbf{\textit{Bound of the spectral gap for the linearized
      Becker-D\"oring operator}}]
  \label{thm:spectral-gap}
  Assume Hypothesis \ref{hyp:linear-section}. Recall the definition of
  $B$ given in \eqref{eq:def-B}:
  \begin{equation*}
  B = \sup_{k \geq 1} \left(\sum_{j=k+1}^\infty \Q_j \right)
  \left(\sum_{j=1}^{k}\dfrac{1}{a_j \Q_j}\right).
\end{equation*}
  It holds that
  \begin{equation}
    \label{eq:spectral-gap-bound-below}
    \frac{1}{4B} \leq \lambda_0,
  \end{equation}
  which should be understood as saying that $B = +\infty$ if
  $\lambda_0=0$. If we additionally assume that
  \begin{equation}
    \label{eq:def-M3}
    M_3 :=
    \sum_{i=1}^\infty
    \frac{1}{a_i \Q_i} \left( \sum_{j=i+1}^{\infty} j \Q_j \right)^2
    < +\infty
  \end{equation}
  and call
  \begin{equation}
    \label{eq:def-M2}
    M_2 := \sum_{i=1}^\infty i^2 \Q_i
  \end{equation}
  (which is finite due to \eqref{eq:A}), then we also have an upper
  bound of $\lambda_0$:
  $$B - \frac{M_3}{M_2} \leq \frac{1}{\lambda_0} \leq 4 B,$$
  which should be understood as saying that $B = +\infty$ if and only
  if $\lambda_0 = 0$.
\end{thm}

Using this it is easy to give quite general conditions on the
coefficients $a_i$, $b_i$ such that the spectral gap $\lambda_0$ is
strictly positive whenever $z < z_\mathrm{s}$. In particular, the
following lemma applies to the ``typical'' coefficients
\eqref{eq:PT-coefficients} and \eqref{eq:CF-coefficients}. Notice that
the following result has already been obtained in
Prop. \ref{prp:L-decomposition} by a completely different argument;
however, the following Corollary is \textit{constructive}, relying on
the above Theorem:
\begin{cor}
  \label{cor:spectral-gap-finite}
  Assume Hypothesis \ref{hyp:linear-section} and conditions
  \eqref{eq:ai-bounded-below}, \eqref{eq:Qi-limit} and
  \eqref{eq:ai-limit}. Then for any $0 < z < z_\mathrm{s}$, the
  spectral gap $\lambda_0$ of $\mathcal{L}$ in $\mathcal{H}$ is
  strictly positive.
\end{cor}

\begin{proof}
  Due to \eqref{eq:spectral-gap-bound-below} it is enough to show that
  $B < +\infty$.
  Let us call $\Q_k := z^k Q_k$ as usual, and
  \begin{equation*}
    m_k := \sum_{j=k+1}^\infty \Q_j,
    \quad
    n_k := \sum_{j=1}^k \frac{1}{a_j \Q_j}
    \fa
    k \geq 1.
  \end{equation*}
  We will show that
  \begin{gather}
    \label{eq:mk-asymp}
    m_k
    \overset{k \to +\infty}{\sim}
    \left( \frac{z_\mathrm{s}}{z} - 1 \right)^{-1} \Q_k \\
    \label{eq:nk-asymp}
    n_k
    \overset{k \to +\infty}{\sim}
    \left( \frac{z_\mathrm{s}}{z} - 1 \right)^{-1}
    \frac{1}{a_k \Q_k},
  \end{gather}
  which then implies that
  \begin{equation*}
    m_k n_k =  \frac{1}{a_k} \frac{m_k}{\Q_k} (n_k a_k \Q_k)
    \leq \frac{1}{\underline{\sigma}} \frac{m_k}{\Q_k} (n_k a_k \Q_k)
    < C
    \quad \text{ for all } k \geq 1
  \end{equation*}
  for some $C > 0$, due to the lower bound on
  \eqref{eq:ai-bounded-below} and the limits \eqref{eq:mk-asymp} and
  \eqref{eq:nk-asymp}. This implies that $B = \sup_{k \geq 1} ( m_k
  n_k ) < C$, so we just have to prove \eqref{eq:mk-asymp} and
  \eqref{eq:nk-asymp}.

  We prove them by using the Stolz-Ces\`aro theorem. Notice that
  \begin{equation*}
    \frac{m_{k+1} - m_k}{\Q_{k+1} - \Q_k} =
    - \frac{\Q_{k+1}}{\Q_{k+1} - \Q_k}
    = \left( \frac{\Q_k}{\Q_{k+1}} - 1 \right)^{-1}
    \longrightarrow \left( \frac{z_\mathrm{s}}{z} - 1 \right)^{-1},
  \end{equation*}
  due to \eqref{eq:Qi-limit}. Since $m_k$ is strictly
  decreasing and tends to $0$ as $k \to +\infty$, this implies
  \eqref{eq:mk-asymp}.

  In a similar way, due to \eqref{eq:Qi-limit} and
  \eqref{eq:ai-limit},
  \begin{equation*}
    \frac{n_{k+1} - n_k}{(a_{k+1}\Q_{k+1})^{-1} - (a_k \Q_k)^{-1}} =
    \frac{(a_{k+1} \Q_{k+1})^{-1}}{(a_{k+1}\Q_{k+1})^{-1} - (a_k \Q_k)^{-1}}
    =
    \left(
      1 - \frac{a_{k+1} \Q_{k+1}}{a_k \Q_k}
    \right)^{-1}
    \longrightarrow
    \left(
      1 - \frac{z}{\zs}
    \right)^{-1},
  \end{equation*}
  which, since $n_k$ is strictly increasing and unbounded, implies
  \eqref{eq:nk-asymp}.
\end{proof}

Let us begin the proof of Theorem \ref{thm:spectral-gap}. We will
prove it through an inequality which is stronger than
\eqref{eq:BDgap}, and equivalent in some cases, as we will show
immediately afterwards:

\begin{lem}
  Assume Hypothesis \ref{hyp:linear-section}. Define $\lambda_1$ as
  the largest nonnegative constant such that
  \begin{equation}
    \label{eq:BD-Hardy}
    \lambda_1 \sum_{i=1}^\infty \Q_i
    \left(h_i-i\,h_1\right)^2
    \leq
    \sum_{i=1}^\infty
    a_{i}\Q_i \Q_1
    \big(h_{i+1} - h_i - h_1 \big)^2
    \qquad \text{ for all $h = (h_i)_{i \geq 1}$.}
  \end{equation}
  Then
  \begin{equation}
    \label{eq:A-B-relation}
    \dfrac{1}{4B} \leq \lambda_1 \leq \dfrac{1}{B},
  \end{equation}
  which again should be understood as saying that $\lambda_1 = 0$ if
  and only if $B = +\infty$.
\end{lem}

\begin{proof}
  We call
  \begin{equation*}
    \mu_i := \Q_{i+1},
    \quad \nu_i=a_{i}\Q_{i}
    \qquad \text{ for } i \geq 1.
  \end{equation*}
  Hardy's inequality from Appendix \ref{sec:Hardy} says that $B$ is
  finite if and only if the following inequality holds for some
  $\lambda_1 > 0$:
  \begin{equation}
    \label{eq:HP1}
    \lambda_1 \sum_{i=1}^\infty
    \mu_i \left(\sum_{j=1}^i f_j\right)^2
    \leq
    \sum_{i=1}^\infty \nu_i\,f_i^2
    \qquad \text{ for all $f = (f_i)_{i \geq 1}$.}
  \end{equation}
  In addition, if any of these conditions hold, we have
  \eqref{eq:A-B-relation}. Hence it is enough to show that
  \eqref{eq:BD-Hardy} is equivalent to \eqref{eq:HP1}.

  For one implication, assume that \eqref{eq:HP1} holds. If we
  take any sequence $h = (h_i)_{i \geq 1}$ and call
  \begin{equation*}
    f_i:= h_{i+1}-h_{i}-h_1
    \qquad \text{ for } i \geq 1
  \end{equation*}
  then \eqref{eq:HP1} is just \eqref{eq:BD-Hardy}: notice that the
  term on the right hand side of \eqref{eq:HP1} is
  \begin{equation}
    \label{eq:HP2}
    \sum_{i=1}^\infty \nu_i\,f_i^2
    = \sum_{i=1}^\infty a_i\,\Q_i\left(h_{i+1}-h_i-h_1\right)^2
  \end{equation}
  and the term to the left of \eqref{eq:HP1} is
  \begin{multline}
    \label{eq:HP3}
    \lambda_1 \sum_{i=1}^\infty
    \mu_i \left(\sum_{j=1}^i f_j\right)^2
    =
    \lambda_1 \sum_{i=1}^\infty
    \mu_i \left(\sum_{j=1}^i \left(h_{j+1}-h_{j}-h_1\right) \right)^2
    \\
    =
    \lambda_1 \sum_{i=1}^\infty
    \mu_i \left( h_{i+1}-(i+1)\,h_1 \right)^2
    =
    \lambda_1 \sum_{i=1}^\infty
    \Q_i \left( h_{i}-i\,h_1 \right)^2,
  \end{multline}
  since the first term in this sum is equal to $0$.

  For the reverse implication, assume \eqref{eq:BD-Hardy}. Then for
  any sequence $f = (f_i)_{i \geq 1}$, define recursively
  \begin{equation*}
    h_1 = 0,
    \qquad
    h_i := f_i + h_{i-1} + h_1
    \qquad \text{ for } i \geq 2.
  \end{equation*}
  Since \eqref{eq:HP2} and \eqref{eq:HP3} still hold,
  \eqref{eq:BD-Hardy} implies \eqref{eq:HP1}, and we have finished the
  proof.
\end{proof}

\begin{lem}
  \label{lem:Hardy-Poincare-equivalence}
  Assume Hypothesis \ref{hyp:linear-section}. There is the following
  relationship between inequality \eqref{eq:BD-Hardy} and the spectral
  gap inequality \eqref{eq:BDgap}:
  \begin{equation}
    \label{eq:H-P}
    \frac{1}{\lambda_0}
    \leq
    \frac{1}{\lambda_1}
    \leq
    \frac{1}{\lambda_0} + \frac{M_3}{M_2},
  \end{equation}
  where for the right inequality we need to assume additionally that
  both $M_3$ is finite. (We recall that $M_2$ and $M_3$ were defined
  in \eqref{eq:def-M3}--\eqref{eq:def-M2}).
\end{lem}

\begin{proof}
  Denote the right-hand side of \eqref{eq:BD-Hardy} by $D(h)$ for
  convenience (it is equal to $\mathscr{E}(h,h) = -\ap{\L h,h}$ for
  any $h \in \mathscr{D}(\L)$, but we rename it to include any
  sequence $h$). Let us show the first inequality. For any sequence $h
  = (h_i)_{i \geq 1} \in \ell_{00}$ with $\sum_i i \Q_i h_i = 0$ we
  have
  \begin{equation*}
    \sum_{i=1}^\infty \Q_i
    \left(h_i-i\,h_1\right)^2
    =
    \sum_{i=1}^\infty \Q_i h_i^2
    +
    h_1^2 \sum_{i=1}^\infty i^2 \Q_i
    - 2 h_1 \sum_{i=1}^\infty i \Q_i h_i
    \geq
    \sum_{i=1}^\infty \Q_i h_i^2
    = \|h\|_\H^2,
  \end{equation*}
  since the mixed term obtained when expanding the square vanishes due
  to the orthogonality condition on $h$. We have, using
  \eqref{eq:BD-Hardy},
  \begin{equation*}
    \lambda_1 \|h\|_\H^2
    \leq
    \lambda_1
    \sum_{i=1}^\infty \Q_i
    \left(h_i-i\,h_1\right)^2
    \leq
    D(h).
  \end{equation*}
  This shows the first inequality.

  For the second inequality we assume that $\lambda_0 > 0$ (since
  otherwise it is a trivial statement). Take any sequence $h =
  (h_i)_{i \geq 1} \in \ell_{00}$.  Calling
  \begin{equation*}
    M_2 := \sum_{i=1}^\infty i^2 \Q_i,
    \qquad
    \overline{h} :=
    \frac{1}{M_2} \sum_{i=1}^\infty i \Q_i h_i
  \end{equation*}
  we have
  \begin{equation}
    \label{eq:PH1}
    \sum_{i=1}^\infty \Q_i
    \left(h_i-i\,h_1\right)^2
    =
    \sum_{i=1}^\infty \Q_i
    \left( h_i-i\,\overline{h} \right)^2
    +
    \sum_{i=1}^\infty \Q_i
    \left( i\,\overline{h}-i\,h_1 \right)^2,
  \end{equation}
  since the mixed term obtained from the square is zero due to the
  definition of $\overline{h}$. If we call $g_i :=
  h_i-i\,\overline{h}$ we see that
  $\sum_{i=1}^\infty i \Q_i g_i = 0$
  and we may use \eqref{eq:BDgap} to obtain
  \begin{equation}
    \label{eq:PH2}
    \sum_{i=1}^\infty \Q_i
    \left( h_i-i\,\overline{h} \right)^2
    \leq
    \lambda_0^{-1} D(g)
    = \lambda_0^{-1} D(h).
  \end{equation}

  This deals with the first term in \eqref{eq:PH1}. For the second
  term, use that
  \begin{multline*}
    M_2 \left| \overline{h}- h_1 \right|
    =
    \left| \sum_{i=1}^\infty i \Q_i (h_i - i h_1) \right|
    =
    \left| \sum_{i=1}^\infty i \Q_i
      \sum_{j=1}^{i-1} (h_{j+1} - h_{j} - h_1)
    \right|
    \\
    \leq
    \sum_{i=1}^\infty \sum_{j=1}^{i-1} i \Q_i |h_{j+1} - h_{j} - h_1|
    =
    \sum_{j=1}^\infty |h_{j+1} - h_{j} - h_1| \sum_{i=j+1}^{\infty} i \Q_i
    \\
    \leq
    \left(
      \sum_{j=1}^\infty a_i Q_i (h_{j+1} - h_{j} - h_1)^2
    \right)^{1/2}
    \left(
      \sum_{j=1}^\infty
      \frac{1}{a_j \Q_j} \left( \sum_{i=j+1}^{\infty} i \Q_i \right)^2
    \right)^{1/2}
    = M_3^{1/2} \, D(h)^{1/2}
  \end{multline*}
  to obtain
  \begin{equation}
    \label{eq:PH3}
    \sum_{i=1}^\infty \Q_i
    \left( i\,\overline{h}-i\,h_1 \right)^2
    =
    \left( \overline{h}-h_1 \right)^2
    M_2
    \leq
    \frac{M_3}{M_2}
    D(h).
  \end{equation}
  Using \eqref{eq:PH2} and \eqref{eq:PH3} in \eqref{eq:PH1} we finally
  obtain
  \begin{equation*}
    \sum_{i=1}^\infty \Q_i
    \left(h_i-i\,h_1\right)^2
    \leq
    \left(
      \frac{1}{\lambda_0} + \frac{M_3}{M_2}
    \right)
    D(h)
  \end{equation*}
  for every $h \in \ell_{00}$. It is easy to see that this implies the
  inequality for all sequences $h$, hence proving the second
  inequality in \eqref{eq:H-P}.
\end{proof}

\subsection{Estimates of the spectral gap size near the
  critical density}\label{sub:critical}

In many cases of interest one can estimate the quantity $B$ as $z$
approaches $z_s$. This implies an estimate on the size of the spectral
gap for the linearized Becker-Döring equations, as stated in Theorem
\ref{thm:spectral-gap}. We remark that our main estimate for $B$,
Lemma \ref{lem:B-bounds} below, applies to the coefficients
\eqref{eq:PT-coefficients} as well as \eqref{eq:CF-coefficients}.

Let us set the notation. We call
\begin{equation*}
  g_k := - \log Q_k - k \log z_\mathrm{s}
  \quad \text{ for } k \geq 1,
  \qquad w := \log \frac{z_\mathrm{s}}{z}
\end{equation*}
so that
\begin{equation*}
  \Q_k = Q_k z^k = \exp\left(-g_k - k w\right)
  \quad \text{ for } k \geq 1.
\end{equation*}

\begin{lem}
  \label{lem:B-bounds-prev}
  Assume that there exists $k_0 > 0$ such that:
  \newcounter{saveenum}
  \begin{enumerate}
  \item \label{it:gk-increasing} The sequence $(g_k)_k$ is increasing
    for $k \geq k_0$.
  \item The sequence $(g_k / k)_k$ is decreasing for $k \geq k_0$.
    \setcounter{saveenum}{\value{enumi}}
  \end{enumerate}
  Assume also that:
  \begin{enumerate}
    \setcounter{enumi}{\value{saveenum}}
  \item There are constants $C > 0$ and $0 < \mu < 1$ such that
    \begin{equation}
      \label{eq:g-difference-asymptotics}
      g_{k+1} - g_k \sim C k^{\mu-1}.
    \end{equation}
  \item There are constants $C, \alpha > 0$ such that
    \begin{equation}
      \label{eq:a_k-asymptotics}
      a_k \sim C k^\alpha
      \quad \text{ as } k \to +\infty.
    \end{equation}
  \end{enumerate}
  Then for some numbers $C_1 < C_2$ depending only on the coefficients
  $(a_k)_k$ and $(b_k)_k$, we have that for all $z_\mathrm{s}/2 < z <
  \zs$,
  \begin{equation}
    \label{eq:g-estimateA}
    \frac{C_1 \Q_{k}}{w + k^{\mu-1}}
    \ \leq \
    \sum_{j=k}^\infty \Q_j
    \ \leq \
    \frac{C_2 \Q_{k}}{w + k^{\mu - 1}}
    \qquad \text{ for $k \geq 1$.}
  \end{equation}
  and
  \begin{equation}
    \label{eq:g-estimateB}
    \frac{C_1}{k^\alpha \Q_{k} (w + k^{\mu-1})}
    \ \leq \
    \sum_{j=1}^{k}\dfrac{1}{a_j \Q_j}
    \ \leq \
    \frac{C_2}{k^\alpha \Q_{k} \left(w + k^{\mu-1} \right)}
    \qquad \text{ for $k \geq 1$.}
  \end{equation}
\end{lem}

\begin{rem}
  We give the above estimates only for $z > z_\mathrm{s}/2$ for
  simplicity, and since we are interested in the behavior for $z$
  close to $z_\mathrm{s}$. They can easily be derived for $z >
  \delta$, for any $\delta > 0$ (with constants $C_1$, $C_2$ depending
  on $\delta$), but are very poor estimates for small $z$.
\end{rem}

\begin{rem}
  \label{rem:example-coefs}
  Both example coefficients \eqref{eq:PT-coefficients} and
  \eqref{eq:CF-coefficients} satisfy all hypotheses of Lemma
  \ref{lem:B-bounds-prev}. Notice that for \eqref{eq:CF-coefficients}
  the sequence $g_k$ is
  \begin{equation*}
    g_k =  -\log z_\mathrm{s} + \sigma (k-1)^\mu
    \quad \text{ for } k \geq 1
  \end{equation*}
  while for \eqref{eq:PT-coefficients} we have
  \begin{equation*}
    g_k = \alpha \log k - \log z_\mathrm{s}
    + \sum_{j=2}^k \log \left( 1 + \frac{q/z_\mathrm{s}}{j^{1-\mu}} \right)
    \quad \text{ for } k \geq 1.
  \end{equation*}
  From this it is easy to check that all the requirements of Lemma
  \ref{lem:B-bounds-prev} are met.
\end{rem}

 \begin{rem} Notice that, for any sequence $(g_k)_k$ satisfying the assumptions of Lemma \ref{lem:B-bounds-prev}, the following hold:
  \begin{enumerate}
  \item[(i)] There is a constant $C > 0$ such
    that
    \begin{equation}
      \label{eq:hyp-gk/k-asymptotics}
      g_k / k \leq C k^{\mu - 1}
      \quad \text{ for $k \geq 1$.}
    \end{equation}
    This is a simple consequence of
    \eqref{eq:g-difference-asymptotics}. In particular, $g_k / k$ tends
    to $0$ as $k \to +\infty$.
  \item[(ii)] There exists $C >0$ such that     \begin{equation}
      \label{eq:gk-asymptotics}
      g_k \sim (C/\mu) k^{\mu}
      \quad \text{ as } k \to +\infty.
    \end{equation}
 This is easily deduced from \eqref{eq:g-difference-asymptotics} by using the Stolz-Ces\`aro Theorem.  In particular,
    $g_k$ diverges to $+\infty$.
  \item[(iii)] Points (1), (2) together with
    \eqref{eq:hyp-gk/k-asymptotics} and \eqref{eq:gk-asymptotics},
    show that there exists $k_1 > 0$  such that, for all $k \geq k_1$,
    \begin{gather}
      \label{eq:gk-increasing-comparison}
      g_k \geq g_j \quad \text{ for all } j \leq k,
      \\
      \label{eq:gk/k-decreasing-comparison}
      \frac{g_k}{k} \leq \frac{g_j}{j} \quad \text{ for all } j \leq k.
    \end{gather}
  \end{enumerate}
  All this properties will be used repeatedly in the proof of Lemma \ref{lem:B-bounds-prev}.
 \end{rem}

\begin{proof}[Proof of Lemma \ref{lem:B-bounds-prev}]
  We will use in this proof a simple bound for the following geometric
  sum: for any $\xi > 0$
  \begin{equation}
    \label{eq:geometric}
    \frac{\exp(-k \xi) }{\xi}
    \leq
    \sum_{j=k}^\infty \exp(-k \xi)
    \leq
    \frac{\exp(-(k-1) \xi)}{\xi}
    \quad \text{ for } k \geq 0.
  \end{equation}
  Also, we will denote by $C$ any constant that only depends on the
  coefficients $(a_k)_k$ and $(b_k)_k$, and which may change from one
  line to the next. We will show \eqref{eq:g-estimateA} and \eqref{eq:g-estimateB} for
  $k \geq k_1$. Note that this is enough, since for $k$ bounded there
  are always constants $C_1$, $C_2$ which satisfy
  \eqref{eq:g-estimateA}--\eqref{eq:g-estimateB}.

  \step{Step 1: Proof of the upper bound in (\ref{eq:g-estimateA}).}
  We use two different estimates to prove this upper bound. The first
  one is
  \begin{equation}
    \label{eq:ua1}
    \sum_{j=k}^\infty \Q_j
    = \sum_{j=k}^\infty \exp(-g_j - j w)
    \leq \exp(-g_k) \sum_{j=k}^\infty \exp(- k w)
    \leq \Q_k \frac{\exp(w)}{w} \leq C \frac{\Q_k}{w},
  \end{equation}
  using that $g_k$ is increasing, and also
  Eq. \eqref{eq:geometric}. The second one is
  \begin{equation}
    \label{eq:ua2}
    \sum_{j=k}^\infty \Q_j
    = \sum_{j=k}^\infty \exp(-g_j - j w)
    \leq
    \exp(-kw) \sum_{j=k}^\infty \exp(-g_j)
    \leq C \Q_k k^{1-\mu},
  \end{equation}
  due to Lemma B. \ref{lem:exp-sum-asymptotics} (see also the Remark B. 1
  immediately afterwards). We deduce from \eqref{eq:ua1} and
  \eqref{eq:ua2} that $\sum_{j=k}^\infty \Q_j$ is bounded by the
  harmonic mean of both right hand sides, this is,
  \begin{equation}
    \label{eq:ua3}
    \sum_{j=k}^\infty \Q_j
    \leq \frac{C \Q_k}{w + k^{\mu-1}}.
  \end{equation}

  \step{Step 2: Proof of the lower bound in (\ref{eq:g-estimateA}).}

  Using that $(g_k/k)_k$ is decreasing we have
  \begin{equation}\begin{split}
    \label{eq:ua4}
    \sum_{j=k}^\infty\exp(-g_j - j w)
    &= \sum_{j=k}^\infty \exp\left(-j (g_j/j + w)\right)\\
    &\geq \sum_{j=k}^\infty \exp\left(-j (g_k/k + w)\right)
    \geq \frac{\Q_k}{w + g_k/k}
    \geq \frac{C \Q_k}{w + k^{\mu - 1}}.
  \end{split}\end{equation}
  Here we have also used \eqref{eq:hyp-gk/k-asymptotics}. \\

  \step{Step 3: Proof of the upper bound in (\ref{eq:g-estimateB}).}

  We use two different bounds, in a similar way as in Step 1. The
  first one is
  \begin{equation*}
    \sum_{j=1}^k \Q_j^{-1} a_j^{-1}
    =
    \sum_{j=1}^k \exp(g_j + j w) j^{-\alpha}
    \leq
    \exp(g_k) \sum_{j=1}^k \exp(j w) j^{-\alpha}
  \end{equation*}
  where we have used \eqref{eq:gk-increasing-comparison}. Now, since
  $\exp(j w) \leq \exp((y+1)w)$ and $y^{-\alpha} \geq j^{-\alpha}$ for
  any $y \in (j-1,j)$, one has
  $$\sum_{j=1}^k \exp(j w) j^{-\alpha} \leq \sum_{j=1}^k \int_{j-1}^j \exp((y+1) w) y^{-\alpha}\dy=\int_0^k \exp((y+1) w) y^{-\alpha}\dy.$$
Using then Lemma B. \ref{lem:incomplete-gamma-estimate-2} in Appendix B (with $\mu=1$), we get
$\sum_{j=1}^k \exp(j w) j^{-\alpha} \leq C \exp(k w) k^{-\alpha}$ from which we deduce that
  \begin{equation*}
    \sum_{j=1}^{k} (\Q_j a_j)^{-1}
    \leq
\dfrac{C}{w\Q_k k^\alpha} \qquad k \geq 1.
  \end{equation*}

  The second bound
  is
  \begin{equation*}
    \sum_{j=1}^k \exp(g_j + j w)  j^{-\alpha}
    \leq
    \exp(kw) \sum_{j=1}^k \exp(g_j) j^{-\alpha}
    \leq C \Q_k^{-1} k^{1-\mu-\alpha},
  \end{equation*}
  where we used Lemma B. \ref{lem:exp-sum-asymptotics} again. Taking the
  harmonic mean of these two bounds gives the desired one as before.

  \step{Step 4: Proof of the lower bound in (\ref{eq:g-estimateB}).}

  Similarly to \eqref{eq:ua4} we have
  \begin{equation*}\begin{split}
    \sum_{j=1}^k \Q_j^{-1} a_j^{-1}
    &=
    \sum_{j=1}^k \exp(g_j + j w) j^{-\alpha}
    \geq
    k^{-\alpha} \sum_{j=1}^k \exp\left(j (g_k/k + w)\right)\\
    &\geq
    k^{-\alpha} \Q_k^{-1} \frac{1}{w + g_k/k}
    \geq
    k^{-\alpha} \Q_k^{-1} \frac{C}{w + k^{\mu - 1}},\end{split}
  \end{equation*}
  using again \eqref{eq:hyp-gk/k-asymptotics} and \eqref{eq:gk/k-decreasing-comparison}.
\end{proof}

As a consequence of the previous lemma we have an estimate of the size
of the quantity $B$ in \eqref{eq:def-B}, useful mainly because it
gives its blow-up rate as $z$ approaches $z_\mathrm{s}$, making more precise the above Lemma \ref{lem:deltak}:

\begin{lem}
  \label{lem:B-bounds}
  Assume the conditions in Lemma \ref{lem:B-bounds-prev}. If $\alpha >
  2(1-\mu)$ there are numbers $C_1, C_2 > 0$ such that
  \begin{equation*}
    C_1 \left(\log\frac{z_\mathrm{s}}{z}\right)^{-2 + \frac{\alpha}{1-\mu}}
    \leq
    B
    \leq
    C_2 \left(\log\frac{z_\mathrm{s}}{z}\right)^{-2 + \frac{\alpha}{1-\mu}}
  \end{equation*}
  for $z < z_\mathrm{s}$. In the case $\alpha \geq 2(1-\mu)$ the
  quantity $B$ is bounded uniformly up to $z = z_\mathrm{s}$: there
  exist numbers $C_1, C_2 > 0$ such that
  \begin{equation*}
    C_1
    \leq
    B
    \leq
    C_2
  \end{equation*}
  for $z \leq z_\mathrm{s}$.
\end{lem}

\begin{proof} Recall the definition of
  $B$ given in \eqref{eq:def-B}:
  \begin{equation*}
  B = \sup_{k \geq 1} \left(\sum_{j=k+1}^\infty \Q_j \right)
  \left(\sum_{j=1}^{k}\dfrac{1}{a_j \Q_j}\right).
\end{equation*}
  Using the bounds in Lemma \ref{lem:B-bounds-prev} and the fact that
  $\Q_{k+1}/\Q_k$ is bounded from below and above uniformly with
  respect to $w$, we have that, for some numbers $C_1$, $C_2$,
  \begin{equation}
    \label{eq:Bb1}
    C_1 \sup_{k\geq 1} \frac{1}{k^\alpha \left(w + k^{\mu-1} \right)^2}
    \leq
    B
    \leq
    C_2 \sup_{k\geq 1} \frac{1}{k^\alpha \left(w + k^{\mu-1} \right)^2}.
  \end{equation}
  In the case $\alpha < 2(1-\mu)$, the minimum of the function $x
  \mapsto x^\alpha (w + x^{\mu - 1})^2$ occurs at $x = C_{\alpha,\mu}
  w^{-1/(1-\mu)}$, for $C_{\alpha,\mu}$ some quantity depending only
  on $\alpha$ and $\mu$. Hence we obtain
  \begin{equation*}
    C_1 w^{-2 + \frac{\alpha}{1-\mu}}
    \leq
    B
    \leq
    C_2 w^{-2 + \frac{\alpha}{1-\mu}},
  \end{equation*}
  which gives the result when $\alpha < 2(1-\mu)$ (recall that $w =
  \log(z_\mathrm{s}/z)$).

  On the other hand, if $\alpha \geq 2(1-\mu)$, then the supremum in
  \eqref{eq:Bb1} is bounded for any value of $w \geq 0$, which shows
  the second part of the result.
\end{proof}

\section{Spectral gap in weighted $\ell^1$ spaces}\label{sec:spectralL1}

We now extend the above spectral gap result to the larger functional
space $X = \ell^1(e^{\eta i} \Q_i)$ given by
\begin{equation}
  \label{eq:def-X}
  X := \Big\{
  h = (h_i)_i \with \|h\|_X := \sum_i \exp\left(\eta i\right) \Q_i|h_i| < +\infty
  \Big\},
\end{equation}
for some $0 < \eta < \frac{1}{2} \log \frac{z_\mathrm{s}}{z}$. The
choice of $\eta$ ensures that this space is actually larger than the
Hilbert space $\H = \ell^2(\Q)$ used in the previous section, since
\begin{equation}
  \label{eq:weight-admissible}
  \|h\|_X
  \leq
  \left( \sum_{i=1}^\infty h_i^2 \Q_i \right)^{1/2}
  \left( \sum_{i=1}^\infty \Q_i \exp\left(2 \eta i\right) \right)^{1/2}
  =\sqrt{2}
  \left( \sum_{i=1}^\infty \Q_i \exp\left(2 \eta i\right) \right)^{1/2}
  \, \norm{h}_\H
  ,
\end{equation}
and the latter parenthesis is finite for $\eta < \frac{1}{2} \log
\frac{z_\mathrm{s}}{z}$. In order to see this notice that, since $\Q_i
= z^i Q_i$,
\begin{equation*}
  \sum_{i=1}^\infty \Q_i \exp\left(2 \eta i\right) =
  \sum_{i=1}^\infty Q_i \,r^i
  \quad \text{ where } \quad
  r := \exp\big(2 \eta + \log z \big).
\end{equation*}
Since the radius of convergence of the power series with coefficients
$i\,Q_i$ is, by definition, $z_\mathrm{s}$, the above sum is finite
whenever $r < z_\mathrm{s}$, this is, when $\eta < \frac{1}{2} \log
\frac{z_\mathrm{s}}{z}$.

Our approach to show that $L$ can be extended to an unbounded operator
on $X$ that has a positive spectral gap uses recent techniques
\cite{GMM} based upon a suitable decomposition of the linearized
operator into a dissipative part and a ``regularizing'' part. We use
the following result, which is a slight improvement over some of the
consequences of \cite{GMM}:

\begin{thm}
  \label{thm:GMM}
Let $\mathcal{Y}$, $\mathcal{Z}$ be two Banach spaces with norms $\|\cdot\|_\mathcal{Y}$,
  $\|\cdot\|_\mathcal{Z}$ such that $\mathcal{Y} \subseteq \mathcal{Z}$ and
  \begin{equation*}
    \|h\|_\mathcal{Z} \leq C_{\mathcal{Y}} \|h\|_\mathcal{Y}, \qquad h \in \mathcal{Y}
  \end{equation*}
  for some $C_\mathcal{Y} > 0$. Let $L_\mathcal{Y}: \D(L_\mathcal{Y}) \to \mathcal{Y}$ and $L_\mathcal{Z}: \D(L_\mathcal{Z})
  \to \mathcal{Z}$ be unbounded operators in $\mathcal{Y}$ and $\mathcal{Z}$, respectively, and such
  that $L_\mathcal{Z}$ is an extension of $L_\mathcal{Y}$ (i.e., $\D(L_\mathcal{Y}) \subseteq
  \D(L_\mathcal{Z})$ and $L_\mathcal{Z}\vert_{\D(L_\mathcal{Y})} = L_\mathcal{Y}$). Assume that:
  \begin{enumerate}
  \item $L_\mathcal{Y}$ generates a $C_0$-semigroup $(S_t)_{t \geq 0}$ on $\mathcal{Y}$.
  \item It holds that
    \begin{equation}
      \label{eq:s1}
      \| S_t h_0 \|_\mathcal{Y} \leq C_1 \exp\left(-\lambda_1 t\right)  \|h_0 \|_\mathcal{Y},
      \qquad h_0 \in \mathcal{Y},\:\:t \geq 0
    \end{equation}
    for some $C_1 > 0$, $\lambda_1 \in \R$.
  \item $L_\mathcal{Z} = \mathbf{A}  + \mathbf{B}$, where $\mathbf{A}$, $\mathbf{B}$ are two linear operators on
    $\mathcal{Z}$ such that:
    \begin{enumerate}
    \item $\mathbf{A}:\mathcal{Z} \to \mathcal{Y}$ is bounded, i.e.,
      \begin{equation}
        \label{eq:A-bounded}
        \| \mathbf{A} h \|_\mathcal{Y} \leq C_\mathbf{A} \|h\|_\mathcal{Z},
        \qquad h \in \mathcal{Z},
      \end{equation}
      for some $C_\mathbf{A} > 0$.

    \item $\mathbf{B}$ generates a $C_0$-semigroup $(S_t^\mathbf{B})_{t \geq 0}$ on $\mathcal{Z}$ which satisfies
      \begin{equation}
        \label{eq:sB}
        \| S_t^\mathbf{B} h_0 \|_\mathcal{Z} \leq C_2 \exp\left(-\lambda_2 t\right)  \|h_0\|_\mathcal{Z},
        \qquad h_0 \in \mathcal{Z}, \:\:t \geq 0
      \end{equation}
      for some $C_2 > 0$, and some $\lambda_2 > \lambda_1$.

    \end{enumerate}
  \end{enumerate}
  Then $L_\mathcal{Z}$ generates a $C_0$-semigroup $(V_t)_{t \geq 0}$
  on $\mathcal{Z}$ (an extension of $(S_t)_{t \geq 0}$) which
  satisfies
  \begin{equation}
    \label{eq:s3}
    \| V_t h_0 \|_\mathcal{Z} \leq C \exp\left(-\lambda_1 t\right) \|h_0\|_\mathcal{Z},
    \qquad h_0 \in \mathcal{Z}, \:\: t \geq 0
  \end{equation}
  for $C = C_2 + C_\mathcal{Y} C_1 C_2 C_\mathbf{A} (\lambda_2 - \lambda_1)^{-1}$.
\end{thm}

This theorem is directly inspired from the results in
\cite{GMM}. However, while the proofs in \cite{GMM} are based on
estimates of the resolvents of $L_\mathcal{Y}$ and $L_\mathcal{Z}$, we present a
completely different one based on the study of the dynamics generated
by $L_\mathcal{Y}$ and $L_\mathcal{Z}$. In addition to being remarkably simple, this
proof allows us to state the result in general for any two Banach
spaces (not necessarily one of them Hilbert) and to reach an
exponential decay like $\lambda_1$ in \eqref{eq:s3}. The resolvent
methods in \cite{GMM} are able to give more precise estimates on the
behavior of the semigroup on $\mathcal{Z}$, but here \eqref{eq:s3} will suffice.

\begin{proof}
  We notice first that $\mathbf{A}\::\:\mathcal{Z} \to \mathcal{Z}$ is
  a bounded operator. Since $\mathbf{B}$ is the generator of a
  $C_0$-semigroup $(S_t^\mathbf{B})_{t \geq 0}$ in $\mathcal{Z}$, by
  the bounded perturbation theorem, the operator $L_\mathcal{Z} =
  \mathbf{A} + \mathbf{B}$ generates a $C_0$-semigroup $(V_t)_{t \geq 0}$ in
  $\mathcal{Z}$ which, additionally, is given by the Duhamel's formula
  \cite[Theorem 4.9]{Bana}:
  \begin{equation*}\label{eq:Duhamel}
    V_t h_0=S_t^\mathbf{B}\,h_0 + \int_0^t
    V_{t-s}\left(\mathbf{A}\,S_s^\mathbf{B} h_0\right)\d s \qquad
    \forall h_0 \in \mathcal{Z},\:t\geq 0.
  \end{equation*}
  Thus, for fixed $h_0 \in \mathcal{Z}$ and $t \geq 0$, we  have
  \begin{multline}\label{eq:Vth}
    \|V_t h_0\|_\mathcal{Z}  \leq \|S_t^\mathbf{B} h_0\|_{\mathcal{Z}}
    + \int_0^t \|V_{t-s}\left(\mathbf{A}\,S_s^\mathbf{B}
    h_0\right)\|_\mathcal{Z}\d s
    \\
    \leq C_2 \exp\left(-\lambda_2 t\right) \|h_0\|_\mathcal{Z}
    +  \int_0^t \|V_{t-s}\left(\mathbf{A}\,S_s^\mathbf{B} h_0\right)
    \|_\mathcal{Z}\d s,
  \end{multline}
  where we used \eqref{eq:sB} in the last inequality. Now, for any $s
  \in (0,t)$, $\left(\mathbf{A}\,S_s^\mathbf{B} h_0\right) \in
  \mathcal{Y}$ since $\mathbf{A}$ maps $\mathcal{Z}$ to
  $\mathcal{Y}$. Then, since $L_\mathcal{Z}$ is an extension of
  $L_\mathcal{Y}$, it is clear that $(V_t)_{t \geq 0}$ is an extension
  of $(S_t)_t$ and $$V_{t-s}\left(\mathbf{A}\,S_s^\mathbf{B}
    h_0\right)=S_{t-s}\left(\mathbf{A}\,S_s^\mathbf{B} h_0\right) \in
  \mathcal{Y} \qquad \forall s \in (0,t).$$ In particular, thanks to
  \eqref{eq:s1},
\begin{multline*}
 \int_0^t  \|V_{t-s}\left(\mathbf{A}\,S_s^\mathbf{B} h_0\right)\|_\mathcal{Z}\d s \leq
    C_\mathcal{Y} \int_0^t \| S_{t-s}(\mathbf{A} S_s^\mathbf{B} h_0) \|_\mathcal{Y} \ds \\
    \leq
    C_\mathcal{Y} C_1
    \int_0^t \exp\left(-\lambda_1 (t-s)\right) \| \mathbf{A} S_s^\mathbf{B} h_0 \|_\mathcal{Y} \ds
    \\
    \leq
    C_\mathcal{Y} C_1 C_\mathbf{A}
    \int_0^t \exp\left(-\lambda_1 (t-s)\right) \| S_s^\mathbf{B} h_0 \|_\mathcal{Z} \ds.
\end{multline*}
Using again \eqref{eq:sB} we deduce from \eqref{eq:Vth}
\begin{multline*}\|V_t h_0\|_\mathcal{Z}  \leq   C_2  \exp(-\lambda_2 t)\|h_0\|_{\mathcal{Z}} +
    C_\mathcal{Y} C_1 C_2 C_\mathbf{A} \,e^{-\lambda_1 t} \,\|h_0\|_\mathcal{Z}
    \int_0^t \exp\left(-(\lambda_2 - \lambda_1) s\right) \ds \\
    \leq C_2 \exp(-\lambda_2 t)\|h_0\|_{\mathcal{Z}} +
    \frac{C_\mathcal{Y} C_1 C_2 C_\mathbf{A}}{\lambda_2 - \lambda_1}
    \,\exp\left(-\lambda_1 t\right) \,\|h_0\|_\mathcal{Z}
  \end{multline*}
which, since $\lambda_2 > \lambda_1$ reads
  \begin{equation*}
   \|V_t h_0\|_\mathcal{Z}
    \leq \left(
      \frac{C_\mathcal{Y} C_1 C_2 C_\mathbf{A}}{\lambda_2 - \lambda_1}
      + C_2
    \right)
   \exp\left(-\lambda_1 t\right) \|h_0\|_\mathcal{Z}, \qquad t \geq 0
  \end{equation*}
 and yields the desired result.
\end{proof}

\bigskip Our motivation for studying the spectral properties of the
linearized operator in the larger space $X$ is that in this larger
space we have useful bounds of the nonlinear remainder term
$\Gamma(h,h)$ defined in \eqref{eq:def-Gamma}:

\begin{prp}
  \label{lem:Gamma-bound-DCF-l1}
  Assume Hypothesis \ref{hyp:linear-section} and take $0 < \eta <
  \frac{1}{2} \log \frac{z_\mathrm{s}}{z}$.  There is a constant $C >
  0$ depending only on $(a_i)_{i \geq 1}$, $(b_i)_{i \geq 1}$ and $z$
  such that
  \begin{equation*}
    \| \Gamma(h,h) \|_X
    \leq
    C \|h\|_X \|h\|_{X_1} \qquad \forall h \in X_1
  \end{equation*}
  where $X_1 = \ell^1((1+\sigma_i) \exp(\eta i) \Q_i)$, i.e.
  \begin{equation}
    \label{eq:def-Xlambda}
    X_1 := \Big\{
    h = (h_i)_i \with \|h\|_{X_1}
    := \sum_i (1+\sigma_i) \exp(\eta i) \Q_i |h_i| < +\infty
    \Big\}
  \end{equation}
  and we recall that $(\sigma_i)_{i\geq 1}$ was defined in \eqref{eq:sigma-i}.
\end{prp}

\begin{proof} Recall that the bilinear operator $\Gamma(f,g)$ is
  defined in weak form by
  \begin{equation*}
    \sum_{i=1}^\infty \Gamma_i(f,g) \Q_i \phi_i
    = \frac{1}{2}\sum_{i \geq 1} a_{i}\Q_i \Q_1
    \big( f_i\,g_1+f_1\,g_i \big)
    \big( \phi_{i+1}-\phi_i-\phi_1 \big)
\end{equation*}
while its strong form is given in \eqref{eq:def-Gamma1}-\eqref{eq:def-Gamma}.  Let us first treat the term $\Gamma_1(h,h)$ (the first component of
  $\Gamma(h,h)$). From \eqref{eq:def-Gamma1} we have
  \begin{multline*}
    \Q_1 | \Gamma_1(h,h) |
    =
    \Q_1 \Big|
      a_1\Q_1 h_1^2
      + \sum_{i=1}^\infty a_i \Q_i h_1 h_i
    \Big|
    \leq
    a_1\Q_1^2 h_1^2
    + |h_1| \sum_{i=1}^\infty a_i \Q_1 \Q_i |h_i|
    \\
    \leq
    \sigma_1 \Q_1 h_1^2
    + |h_1| \sum_{i=1}^\infty \sigma_i \Q_i |h_i|
    \leq
    \frac{2}{\Q_1} \|h\|_X \|h\|_{X_1},
  \end{multline*}
  which bounds $\Gamma_1(h,h)$. For the terms with $i \geq 2$,
  $\Gamma_i(h,h)$ is given by \eqref{eq:def-Gamma}:
  \begin{gather*}
    \Gamma_i(h,h) = \Gamma_i^+(h,h) - \Gamma_i^-(h,h), \quad\text{
      with }
    \\
    \Gamma_i^+(h,h) =
    \frac{a_{i-1}\,\Q_{i-1}\,\Q_1}{\Q_i}\,h_{i-1}\,h_1=b_i\,h_{i-1}\,h_1
    \quad \text{ and } \quad
    \Gamma_i^-(h,h) =a_{i} \Q_1 h_i \qquad (i \geq 2).
  \end{gather*}
  We bound these terms separately. For $\Gamma^+(h,h) =
  (\Gamma_i^+(h,h))_{i \geq 2}$ we have, using \eqref{eq:Qi},
  \begin{multline*}
    \| \Gamma^+(h,h) \|_X =
    \sum_{i=2}^\infty \exp\left(\eta i\right) b_i \Q_i |h_1|\,|h_{i-1}|
    =
    |h_1| \sum_{i=1}^\infty \exp\left(\eta (i+1)\right) b_{i+1} \Q_{i+1} |h_i|
    \\
    =
    \exp\left(\eta\right)\, |h_1|\, \sum_{i=1}^\infty a_i \Q_1  \exp\left(\eta i\right) \Q_i  |h_i|
    \leq
    \exp\left(\eta\right)\, |h_1|\, \sum_{i=1}^\infty \sigma_i \Q_i  \exp\left(\eta i\right) |h_i|
    \leq
    \frac{\exp\left(\eta\right)}{\Q_1} \|h\|_X \, \|h\|_{X_1}.
  \end{multline*}
  The proof is even simpler for $\Gamma^-(h,h) = (\Gamma_i^-(h,h))_{i
    \geq 2}$ since
  \begin{multline*}
    \|\Gamma^-(h,h)\|_X
    = |h_1| \sum_{i=2}^\infty a_i \Q_1 \exp\left(\eta i\right) \Q_i |h_i|
    \leq |h_1| \sum_{i=2}^\infty \sigma_i \exp\left(\eta i\right)\Q_i |h_i|\\
    = |h_1|\, \|h\|_{X_1}
    \leq \frac{1}{\Q_1} \|h\|_X \, \|h\|_{X_1},
  \end{multline*}
  which finishes the result.
\end{proof}

\medskip

The following result shows that we may extend the linearized operator
$L$ from Definition \ref{def:L-compact} to an operator on $X$ with
domain $X_1$. It can be proved by direct estimates on the expression
of $L$ given in \eqref{eq:L} which are similar to those in the proof
of the previous lemma, and we omit its proof:

\begin{lem}
  Assume Hypothesis \ref{hyp:linear-section} and
  \eqref{eq:Qi-limit}. There is a constant $C > 0$ depending only on
  $(a_i)_{i \geq 1}$, $(b_i)_{i \geq 1}$ and $z$ such that
  \begin{equation*}
    \| L(h) \|_X
    \leq
    C \|h\|_{X_1}
    \fa h \in \ell_{00}.
  \end{equation*}
  Similarly, there is another constant $C > 0$ depending only on
  $(a_i)_{i \geq 1}$, $(b_i)_{i \geq 1}$, $z$ and $\eta$ such that
  \begin{equation}
    \label{eq:LX-1-X}
    \sum_{i=1}^\infty \Q_i \frac{\exp(\eta i)}{\sigma_i} | L_i(h) |
    \leq
    C \|h\|_{X}
    \fa h \in \ell_{00}.
  \end{equation}
\end{lem}

The above lemma allows us to extend $L$ to $X_1$:
\begin{defi}
  We denote by $\Lambda$ the extension of the linearized operator $L$
  to the domain $X_1$.
\end{defi}
Recall that the operator  $L$ is given in strong form by
\begin{equation*}
  L_1(h) =
  -\frac{1}{\Q_1}\left(\W_1+ \sum_{k=1}^\infty\W_k\right),
  \qquad
  L_i(h) =
  \frac{1}{\Q_i} \left(\W_{i-1} - \W_{i}\right)
  \qquad (i \geq 2).
\end{equation*}
where $ \W_{i} :=a_{i} \Q_i \Q_1 (h_i + h_1 - h_{i+1}).$ Then, one has

\begin{thm}[\textit{\textbf{Extension of the spectral gap}}]
  \label{thm:spectral-gap-extension}
 Assume Hypothesis \ref{hyp:linear-section} and conditions
  \eqref{eq:ai-bounded-below}, \eqref{eq:Qi-limit} and
  \eqref{eq:ai-limit}, and assume also that $z < \zs$. Take $0 < \eta
  < \frac{1}{2} \log \frac{z_\mathrm{s}}{z}$. Then, the operator
  $\Lambda$ generates a strongly continuous semigroup $(\exp(t
  \Lambda))_{t \geq 0}$ on $X$ and there exists $0 < \lambda_\star
  \leq \lambda_0$ such that, for some $C > 0$,
  \begin{equation*}
    \label{eq:decayX}
    \left\|\exp(t\Lambda)g\right\|_X
    \leq C \exp(-\lambda_\star t) \|g\|_X
    \fa t \geq 0
    \text{ and any } g \in X
    \text{ such that } \sum_{i=1}^\infty i\Q_i g_i=0.
  \end{equation*}
  (We recall that $\lambda_0$ is the size of the spectral gap of
  $\mathbf{L}$ in $\H$, as given by Definition
  \ref{def:BD-spectral-gap}). In addition, if we assume
  \begin{equation}
    \label{eq:lim-ai-infty}
    \lim_{i \to +\infty} a_i = +\infty
  \end{equation}
  then we may take $\lambda_\star = \lambda_0$.
\end{thm}

\begin{proof}  We apply Theorem \ref{thm:GMM} to the Banach spaces
  \begin{equation*}
    \H^\perp := \{ h \in \H \mid \sum_{i=1}^\infty i \Q_i h_i = 0 \},
  \end{equation*}
  a subspace of $\H=\Big\{
 h = (h_i)_i \with
  \|h\|_{\mathcal{H}} < \infty
  \Big\}$ (where we recall that $\|h\|_\H^2 := \frac{1}{2} \sum_{i=1}^\infty\Q_i h_i^2$) and
  \begin{equation*}
    X^\perp := \{ h \in X \mid \sum_{i=1}^\infty i \Q_i h_i = 0 \},
  \end{equation*}
  a subspace of $X= \Big\{
  h = (h_i)_i \with \|h\|_X := \sum_i \exp\left(\eta i\right) \Q_i|h_i| < +\infty
  \Big\}$. Note that under our conditions the map
  \begin{equation*}
    \label{eq:mathbfM}
    \mathbf{M}\::\: h \longmapsto \mathbf{M}(h)
    := \sum_{i=1}^\infty i \Q_i h_i
  \end{equation*}
  is continuous both in $\H$ (as a consequence of \eqref{eq:A}) and in
  $X$; hence, $\H^\perp$ and $X^\perp$ are well-defined closed
  subspaces of $\H$ and $X$, respectively. We define
  \begin{equation*}
    \H_2^\perp := \H^\perp \cap \H_2,
    \qquad
    X_1^\perp := X^\perp \cap X_1
  \end{equation*}
  where $\H_2$ is defined in \eqref{eq:def-DL} and $X_1$ is given by \eqref{eq:def-Xlambda}.
  Let us briefly explain how one can project any sequence $g \in X$ to
  $X^\perp.$ To do so, for any sequence $g=(g_i)_i$, let us introduce
  the sequence
  \begin{equation}
    \label{eq:mathcalM}
    \mathcal{M}_i(g)=\dfrac{1}{\varrho} \mathbf{M}(g)
    \qquad \forall i \geq 1
    \qquad \text{ where } \quad \varrho:=\sum_{k=1}^\infty k \Q_k.
  \end{equation}
  Since $\sum_{i \geq 1} \exp(\eta i)\Q_i < \infty$, one checks that
  $\mathcal{M} \::\:X \to X$ is a bounded operator and for any $g \in
  X$ one easily checks that
  $$g-\mathcal{M}g \in X^\perp.$$
  In the same way, the operator $\mathcal{M}\::\:\H \to \H$ is bounded
  and $h-\mathcal{M}h \in \H^\perp$ for any $h \in \H$.

  To play the role of the linear operators $L_\mathcal{Y}$ and
  $L_\mathcal{Z}$ in Theorem \ref{thm:GMM} we take $\mathbf{L}^\perp
  := \mathbf{L} \vert_{\H^\perp}$ with domain
  $\D(\mathbf{L}^\perp)=\H_2^\perp$ and $\Lambda^\perp := \Lambda
  \vert_{X^\perp}$ with domain $\D(\Lambda^\perp)=X_1^\perp$. From
  \eqref{eq:weak-linear} it is clear that the image of $\mathbf{L}$ is
  contained in $\H^\perp$, and that of $\Lambda$ is contained in
  $X^\perp$, so $\mathbf{L}^\perp$ and $\Lambda^\perp$ are unbounded
  operators on $\H^\perp$ and $X^\perp$, respectively.

  Let us check the hypotheses of Theorem \ref{thm:GMM}. As remarked in
  \eqref{eq:weight-admissible}, $\H^\perp$ is contained in
  $X^\perp$. Moreover, $\H^\perp$ is dense in $X^\perp$, since the set of
  compactly supported sequences with $\sum_{i=1}^\infty i \Q_i h_i =
  0$ is contained in both of them. Points (1) and (2) in Theorem
  \ref{thm:GMM} are established in Corollary
  \ref{cor:non-constructive-gap} where it is shown that one may take $C_1 =
  1$ and $\lambda_1 = \lambda_0$ in Theorem \ref{thm:GMM}, where
  $\lambda_0$ is the spectral gap estimated in Section
  \ref{sec:explicit-gap}.

  It remains to show point (3) of Theorem \ref{thm:GMM}.

  \medskip
  \noindent
  \textbf{Step 1: the splitting $\Lambda^\perp = \mathcal{A} +
    \mathcal{B}$.}
  Take an integer $N \geq 1$, to be fixed later. We first split
  $\Lambda$ as $\mathcal{A} + \mathcal{B}$ in a similar way to
  Proposition \ref{prp:L-decomposition}, and we then write
  $\Lambda^{\perp} = \mathcal{A}^\perp + \mathcal{B}^\perp$ by
  projecting to $X^\perp$. We define $\mathbf{L}^\mathrm{C}$ and
  $\mathbf{L}^\mathrm{M}$ similarly to
  \eqref{eq:def-LC-0}--\eqref{eq:def-LM-0}, but this time in $X$: we
  choose $N > 2$ and set, for all $i \geq 1$,
  \begin{gather}
    \label{eq:def-LC-0-X}
    \Lambda^{\mathrm{C}}_i(h)
    :=
    \sum_{j=1}^\infty \chi_{\{\min\{i,j\} \leq N\}} \, \xi_{i,j} h_j
    - \sigma_i h_i \chi_{\{i \leq N\}}
    \quad \text{ for } h \in X,
    \\
    \label{eq:def-LM-0-X}
    \Lambda^{\mathrm{M}}_i(h)
    :=
    - \sigma_i h_i \chi_{\{i > N\}}
    + \sum_{j=1}^\infty \chi_{\{\min\{i,j\} > N\}} \, \xi_{i,j} h_j
    \quad \text{ for } h \in X_1.
  \end{gather}
  We recall that the notation $\sigma_i$ and $\xi_{i,j}$ was defined
  in \eqref{eq:sigma-i}--\eqref{eq:xi-ij}. Remember in particular that
  \begin{equation}\label{remembersigma}\sigma_i=a_i\Q_1+b_i
    \quad \text{ for } i \geq 2.\end{equation}
  We will see below that the sums converge absolutely for $h$ in the
  corresponding domain, so the definition is correct. We also define
  the truncated identity operator $I^N = (I^N_i)_{i \geq 1}$ as
  \begin{equation*}
    I^N_i(h) := h_i \,\chi_{\{i \leq N\}}
  \end{equation*}
  and then set, for $R > 0$ to be fixed later,
  \begin{align*}
    \mathcal{A} &= \Lambda^{\mathrm{C}} + R I^N,
    \qquad h \in X
    \\
    \mathcal{B} &= \Lambda^{\mathrm{M}} - R I^N,
    \qquad h \in X_1,
  \end{align*}
  so that $\Lambda = \mathcal{A + B}$. We finally project this splitting to
  $X^\perp$ thanks to the operator $\mathcal{M}$ introduced in \eqref{eq:mathcalM}. Namely, let
  \begin{align*}
    \mathcal{A}^\perp\,h&= \mathcal{A} h  - \mathcal{M}(\mathcal{A} h ),
    \qquad h \in X^\perp
    \\
    \mathcal{B}^\perp\,h &= \mathcal{B} h   - \mathcal{M}(\mathcal{B} h ),
    \qquad h \in X_1^\perp,
  \end{align*}
  \medskip
  \noindent
  \textbf{Step 2: $\mathcal{A}$ is ``regularizing''.} We show
  now that $\mathcal{A} h  \in \H$ for all $h \in X$ (which, since $\mathcal{M}\::\:\H \to \H$ is bounded will imply
  then that $\mathcal{A}^\perp(h) \in \H^\perp$ for all $h \in
  X^\perp$). Notice that one may write, more explicitly,
  \begin{equation}
    \begin{cases}
      \label{eq:def-LambdaC}
      \Lambda^{\mathrm{C}}_1(h) = \sum_{j=2}^\infty\xi_{1,j}\,h_j -
      \sigma_1 h_1
      \qquad
      \Lambda^{\mathrm{C}}_2(h)
      = \xi_{2,1}\,h_1 + a_2 \Q_1 \,h_3
      - \sigma_2 h_2
      \\
      \Lambda^{\mathrm{C}}_i(h)=
      \xi_{i,1}h_1 + b_i\,h_{i-1} \,\chi_{\{i \leq
        N+1\}} + a_i \Q_1 \,h_{i+1} \,\chi_{\{i\leq N\}}
      - \sigma_i h_i \chi_{\{i \leq N\}}
      \quad \text{ for } i > 2.
    \end{cases}
  \end{equation}
  The only infinite sum in the definition of $\mathcal{A}$ (the one in
  $\Lambda^\mathrm{C}_1$) converges since, from
  \eqref{eq:xi-12}--\eqref{eq:xi-i1},
  \begin{equation*}
    \big| \sum_{j=2}^\infty\xi_{1,j}\,h_j \big|
    \leq
    \frac{2}{\Q_1} \sum_{j=2}^\infty b_j \Q_j \,|h_j|
    +
    \sum_{j=2}^\infty a_j \Q_j \,|h_j|
    \leq
    \frac{2}{\Q_1} \|h\|_X.
  \end{equation*}
  In order to see that $\mathcal{A}(h) \in \H$ it is enough to
  show that $\sum_{i = N+2}^\infty \Q_i \mathcal{A}_i(h)^2 \leq C
  \|h\|_X^2$ for some $C > 0$ (we omit the first $N+1$ terms, for
  which a similar bound is obviously true). But this is true since for
  $i \geq N+2$ the only term in $\mathcal{A}_i(h)$ is the first one in
  the expression of $\Lambda^\mathrm{C}_i(h)$ in
  \eqref{eq:def-LambdaC} and
  \begin{equation*}
    \sum_{i=N+2}^\infty \Q_i \mathcal{A}_i(h)^2
    =
    h_1^2 \sum_{i=N+2}^\infty \Q_i\, \xi_{i,1}^2
    =
    h_1^2 \sum_{i=N+2}^\infty \Q_i\, (b_i - a_i\Q_1)^2
    \leq A\, h_1^2
    \leq \frac{A}{\Q_1^2} \|h\|_X^2
  \end{equation*}
  where $A$ is defined in \eqref{eq:A}. Thus, $\mathcal{A}:X \to \H$
  is bounded, and hence $\mathcal{A}^\perp:X^\perp \to \H^\perp$ also
  is.

  \medskip
  \noindent
  \textbf{Step 3: $\mathcal{B}$ is strictly dissipative.}
  We prove that
  \begin{equation}
    \label{eq:dissipB1}
    \big\langle \mathrm{sign}(h), \mathcal{B} h \big\rangle_{X',X}
    \leq -\lambda_3 \|h\|_{X}
    \fa h \in X_1,
  \end{equation}
  where $\langle \cdot,\cdot\rangle_{X',X}$ denotes the duality
  pairing between $X$ and its dual $X'$, and $\lambda_3$ will be any
  number larger than $\lambda_0$ when \eqref{eq:lim-ai-infty} holds,
  and some positive number otherwise. One sees that
  \eqref{eq:def-LM-0-X} may be rewritten as
  \begin{equation}
    \label{eq:def-LM-X}
    \Lambda^{\mathrm{M}}_i(h)
    = -\sigma_i\,h_i \chi_{\{i > N\}}+b_i\,h_{i-1} \,\chi_{\{i > N+1\}}
    + a_i\,\Q_1\,h_{i+1} \,\chi_{\{i > N\}}
    \qquad (i \geq 1).
  \end{equation}
  For any $h \in X_1$ one has
  \begin{multline*}
    \ap{ \mathrm{sign}(h), \mathcal{B}h }_{X',X}
    = \sum_{i=1}^\infty \exp\left(\eta i\right) \Q_i \sign(h_i) \mathcal{B}_i(h)
    = - \sum_{i=N+1}^\infty \exp\left(\eta i\right) \sigma_i \Q_i |h_i|
    - R \sum_{i=1}^N \exp\left(\eta i\right) \Q_i |h_i|
    \\
    + \sum_{i = N+1}^\infty \exp\left(\eta i\right) \Q_i \sign(h_i)\, b_i h_{i-1}
    + \sum_{i=N}^\infty \exp\left(\eta i\right) \Q_i \Q_{1} \sign(h_i)\, a_i h_{i+1}.
  \end{multline*}
  The last two sums can be bounded by
  \begin{multline*}
    \sum_{i = N+1}^\infty \exp\left(\eta i\right) \Q_i b_i |h_{i-1}|
    + \sum_{i=N}^\infty \exp\left(\eta i\right) \Q_i \Q_{1} a_i |h_{i+1}|
    \\
    \leq
    \sum_{i=N}^\infty
    \left(\exp\left(\eta \right) \frac{\Q_{i+1}}{\sigma_i \Q_i} b_{i+1}
      + \exp\left(-\eta  \right) \frac{\Q_1 \Q_{i-1}}{\sigma_i \Q_i} a_{i-1} \right)
    \exp\left(\eta i\right) \sigma_i \Q_i |h_i|
    =
    \sum_{i=N}^\infty
    \mu_i \exp\left(\eta i\right) \sigma_i \Q_i |h_i|
  \end{multline*}
  where
  \begin{equation}\label{eq:mui}
    \mu_i :=  \exp\left(\eta  \right)\frac{ \Q_1 a_{i}}{\sigma_i}
    +\exp\left(-\eta  \right) \frac{b_i}{\sigma_i}.
  \end{equation}
  Now, due to \eqref{eq:Qi-limit}, \eqref{eq:ai-limit} and
  \eqref{remembersigma}, $\mu_i$ has a limit as $i \to +\infty$:
  \begin{equation*}
    \frac{ \Q_1 a_{i}}{\sigma_i} =
    \frac{ \Q_1 a_{i}}{\Q_1 a_i + b_i}
    =
    \frac{a_{i}}{a_i + a_{i-1} \frac{\Q_{i-1}}{\Q_i}}
    \longrightarrow \frac{1}{1 + \ell}
  \end{equation*}
  where $\ell := \zs/z \geq 1$, while
  \begin{equation*}
    \frac{ b_i }{\sigma_i}
    =
    \frac{1}{1 + \frac{a_i \Q_1}{b_i}}
    =
    \left( {1 + \frac{a_i \Q_i}{a_{i-1} \Q_{i-1}}} \right)^{-1}
    \longrightarrow
    \frac{\ell}{1 + \ell}.
  \end{equation*}
  Hence,
  \begin{equation*}
    \mu_i \longrightarrow \frac{\exp\left(\eta \right) + \exp\left(-\eta \right) \ell}{1 + \ell}
    \quad \text{ as } i \to +\infty.
  \end{equation*}
  One easily checks that this limit is strictly less than $1$ if and
  only if $\exp\left(\eta  \right) < \ell$. Since we are assuming $\exp\left(\eta  \right) < \sqrt{\ell}
  \leq \ell$, for $N$ large enough we can find $\mu < 1$ such that
  \begin{equation*}
    \ap{ \mathrm{sign}(h), \mathcal{B}h }_{X',X}
    \leq
    - (1-\mu) \sum_{i=N+1}^\infty \exp\left(\eta i\right) \sigma_i \Q_i |h_i|
    - R \sum_{i=1}^N \exp\left(\eta i\right) \Q_i |h_i|.
  \end{equation*}
Choosing $\lambda_3 = (1-\mu) \underline{\sigma}$ and $R =
  \lambda_3$ we obtain
  \begin{equation*}
    \ap{ \mathrm{sign}(h), \mathcal{B}h }_{X',X}
    \leq -\lambda_3 \sum_{i=1}^\infty \exp(\eta i) \Q_i |h_i|.
  \end{equation*}
  This is nothing but \eqref{eq:dissipB1}.

  If we assume additionally that \eqref{eq:lim-ai-infty} holds (i.e.,
  $\lim_{j \to \infty} a_j = +\infty$) then take any $\lambda_3 >
  \lambda_0$. One can choose $N >1$ large enough so that $\sigma_i
  \geq \frac{\lambda_3}{1-\mu}$ for any $i \geq N$ and then, by
  picking $R \geq \lambda_3$ we get
  $$\ap{ \mathrm{sign}(h), \mathcal{B}h }_{X',X}
  \leq -\lambda_3 \sum_{i=1}^\infty e^{\eta i} \Q_i |h_i|,$$ which
  proves \eqref{eq:dissipB1} with any $\lambda_3 > \lambda_0$.

\medskip
  \noindent
  \textbf{Step 4: $\mathcal{B}^\perp$ is strictly dissipative.}
  We now show that if we take any $\lambda_2 < \lambda_3$ (with the
  $\lambda_3$ from Step 3) then one can additionally choose $N$ in
  \eqref{eq:def-LC-0-X}--\eqref{eq:def-LM-0-X} such that
  \begin{equation}
    \label{eq:dissipB1perp}
    \big\langle \mathrm{sign}(h), \mathcal{B}^\perp h \big\rangle_{X',X}
    \leq - \lambda_2 \|h\|_{X}
    \fa h \in X_1^\perp.
  \end{equation}
  This is not obvious, since the property \eqref{eq:dissipB1} of
  $\mathcal{B}$ is not necessarily shared by its projection
  $\mathcal{B}^\perp$. However, in this case and for $N$ large enough,
  $\mathcal{B}^\perp$ happens to be a small perturbation of
  $\mathcal{B}$ on $X_1^\perp$. We need to estimate, for $h \in
  X_1^\perp$, the quantity
  \begin{equation*}
   \mathbf{M}(\mathcal{B} h )
    = \sum_{i=N+1}^\infty i \Q_i \Lambda^{\mathrm{M}}_i(h)
    - R \sum_{i=1}^N i \Q_i h_i.
  \end{equation*}
  We have, with very similar calculations as those needed to obtain
  \eqref{eq:LX-1-X}, that
  \begin{equation}
    \label{eq:p1}
    \sum_{i=N+1}^\infty i \Q_i | \Lambda^{\mathrm{M}}_i(h) |
    \leq
    \epsilon_1(N) \sum_{i=N+1}^\infty
    \Q_i \frac{\exp(\eta i)}{\sigma_i} | \Lambda^{\mathrm{M}}_i(h) |
    \leq
    C \epsilon_1(N) \|h\|_{X}
  \end{equation}
  with $$\epsilon_1(N) := \sup_{i \geq N+1} \frac{i
    \sigma_i}{\exp(\eta i)},$$ which tends to $0$ as $N \to +\infty$
  (as implied by \eqref{eq:Qi-limit} and \eqref{eq:ai-limit} by
  considering the ratio of two consecutive terms in the sequence).  On
  the other hand, since $\sum_{i=1}^\infty i \Q_i h_i = 0$,
  \begin{equation}
    \label{eq:p2}
    R \left| \sum_{i=1}^N i \Q_i h_i \right|
    = R \left| \sum_{i=N+1}^\infty i \Q_i h_i \right|
    \leq
    \frac{R}{\underline{\sigma}}\, \epsilon_1(N) \|h\|_X
  \end{equation}
  where we recall that $\underline{\sigma}=\inf_{i \geq 1}\sigma_i.$
  From \eqref{eq:p1}, \eqref{eq:p2}, \eqref{eq:mathcalM} and
  \eqref{eq:dissipB1} we have
  \begin{equation*}
    \big\langle \mathrm{sign}(h), \mathcal{B}^\perp h \big\rangle_{X',X}
    \leq - \lambda_3 \|h\|_{X}
    + \frac{\sum_{i=1}^\infty \Q_i \exp(\eta i)}{\varrho}
    \left(\frac{R}{\underline{\sigma}}\, + C\right) \epsilon_1(N) \|h\|_X
  \end{equation*}
  for all $h \in X_1^\perp$. Choosing $N$ large enough proves
  \eqref{eq:dissipB1perp}, since $\epsilon_1(N) \to 0$ as $N \to
  +\infty$ (notice that the choice of $R$ from Step 3 is not affected).

\medskip
  \noindent
  \textbf{Step 5: $\mathcal{B}$ generates a semigroup on $X$.}
  To prove that $\mathcal{B} $ is the generator of a $C_0$-semigroup on $X $ we shall invoke Miyadera perturbation Theorem \cite{Bana}. Notice that $\mathcal{B}$ may be written as
$$\mathcal{B}=\mathcal{T} + \mathcal{C}$$
  where $\mathcal{T}$ is the multiplication operator (with domain $\D(\mathcal{T})=\D(\mathcal{B})=X_1$) given by
  $$\mathcal{T}_i(h)=-\left(\sigma_i+R\right)\chi_{\{i \geq N\}}h_i \qquad \forall i \geq 1, \;h \in \D(\mathcal{B})$$
  while $\mathcal{C}$ is defined by
  $$\mathcal{C}_i(h)=b_i h_{i-1}\chi_{\{i > N+1\}}+a_i\,\Q_1\,h_{i+1}\chi_{\{i > N\}}, \qquad i \geq 1.$$
   It is clear that $\mathcal{T}$ is the generator of a $C_0$-semigroup $( \mathbf{U}^t)_{t \geq 0}$ on $X$ given by
$$\mathbf{U}^t_i(h)=\exp\left(-t\left( \sigma_i+R\right)\chi_{\{i \geq N\}}\right)\,h_i \qquad i \geq 1,\:\:t \geq 0,\;\;h \in X.$$
Moreover, $\mathcal{C}\::\:X_1 \to X$ is bounded. Indeed, it is not
difficult to check, as we did in the proof of \eqref{eq:dissipB1} that

\begin{equation}
  \label{eq:Cmui}
  \|\mathcal{C}(h)\|_{X}
  \leq
  \sum_{i=N+1}^\infty \mu_i \,\sigma_i \exp\left(\eta i\right)\Q_i |h_i|
  \leq
  \|h\|_{X_1} \sup_{i \geq N+1}\mu_i
  \qquad \forall h=(h_i)_{i\geq 1} \in X_1
\end{equation}
where $\mu_i$ is defined by \eqref{eq:mui}. Since the sequence
$(\mu_i)_i$ is bounded, this shows in particular, that $\mathcal{C}$
is $\mathcal{T}$-bounded. One deduces then from \eqref{eq:Cmui} that
$$\|\mathcal{C}\mathbf{U}^t h\|_X
\leq
\sum_{i=N+1}^\infty  \,\mu_i\sigma_i
\exp\left(-t(\sigma_i+R)\right)
\exp\left(\eta\, i\right)\Q_i |h_i|
\qquad h \in X_1 \,,\;t \geq 0$$
which readily yields
$$\int_0^1 \|\mathcal{C} \mathbf{U}^t h\|_X \d t \leq \sum_{i=N+1}^\infty \,\mu_i \dfrac{\sigma_i}{\sigma_i+R} \exp\left(\eta\, i\right)\,\Q_i |h_i|.$$
In particular, we see that
$$\int_0^1 \|\mathcal{C}\,\mathbf{U}^t h \|_X \d t \leq \mu \|h\|_X \qquad \forall h \in X_1^\perp$$
where $\mu:=\sup_{i \geq N+1} \mu_i$. We already saw in Step 3 that
$N$ can be chosen large enough so that $\mu < 1.$ Then, the above
inequality exactly means that $\mathcal{C}$ is a Miyadera perturbation
of $\mathcal{T}$ (see \cite[Section 4.4, p. 127-128]{Bana}) and that
$\mathcal{B}=\mathcal{T}+\mathcal{C}$ is the generator of a
$C_0$-semigroup $(S_t^\mathcal{B})_{t \geq 0}$ on $X.$ Notice then
that \eqref{eq:dissipB1} exactly means that $(\lambda_3+\mathcal{B})$
is the generator of a contraction semigroups in $X$.

\medskip
\noindent
\textbf{Step 6: $(\lambda_2+\mathcal{B}^\perp)$ generates a
  $C_0$-semigroup of contractions in $X^\perp$} Notice that the
approach used in the previous step seems difficult to apply directly
to $\mathcal{B}^\perp$, the reason being essentially that, in the
above splitting $\mathcal{B}=\mathcal{T}+\mathcal{C}$, the operators
$\mathcal{T}$ and $\mathcal{C}$ does not map their respective domains
to $X^\perp$. To prove that $(\lambda_2+\mathcal{B}^\perp)$ generates
a $C_0$-semigroup of contractions in $X^\perp$, we adopt another more
indirect way. Since $\mathcal{A}\::\: X \to \H$ is bounded, it is
clear that $\mathcal{A}\::\;X \to X$ is a bounded operator. Therefore,
by the bounded perturbation theorem, $\Lambda=\mathcal{A}+\mathcal{B}$
is the generator of a $C_0$ semigroup $(T_t)_{t \geq 0}$ in
$X$. Moreover, since $\Lambda (X_1) \subset X^\perp$, the closed
subspace $X^\perp \subset X$ is invariant under $(T_t)_{t \geq
  0}$. There, as well-known \cite[Chapter II.2.3, p. 60-61]{engel},
the restriction $\Lambda^\perp=\Lambda\vert_{X^\perp}$ with domain
$\D(\Lambda^\perp)=\D(\Lambda) \cap X^\perp=X_1^\perp$ is the
generator of a $C_0$-semigroup in $X^\perp$. We already saw that
$\Lambda^\perp=\mathcal{A}^\perp+\mathcal{B}^\perp$ where
$\mathcal{B}^\perp$ is a bounded operator with domain $X_1^\perp$ and
$\mathcal{A}^\perp\::\:X^\perp \to \H^\perp$ is bounded. In
particular, $\mathcal{A}^\perp\::\:X \to X$ is also a bounded operator
and then,
$$\mathcal{B}^\perp=-\mathcal{A}^\perp+\Lambda^\perp$$
is the bounded perturbation of the generator $\Lambda^\perp.$ Thus,
$\mathcal{B}^\perp$ is the generator of a $C_0$-semigroup
$(S_t^{\mathcal{B}^\perp})_t$ in $X^\perp.$ Moreover, according to
Step 4, inequality \eqref{eq:dissipB1perp} implies that
$\lambda_2+\mathcal{B}^\perp$ is a dissipative operator in $X^\perp$
(in the sense of \cite[Chapter II.3.b, p. 82]{engel}). Since
$\lambda_2+\mathcal{B}^\perp$ generates a $C_0$-semigroup in
$X^\perp$, according to the Lumer-Phillips Theorem, this semigroup is
a contraction semigroup, or equivalently
$$\|S_t^{\mathcal{B}^\perp} h_0 \|_X \leq \exp(-\lambda_2 t)\|h_0\|_X \qquad \forall h_0 \in X^\perp, \quad t \geq 0$$
i.e.  \eqref{eq:sB} holds with $C_2=1$ and $\lambda_2$ provided by \eqref{eq:dissipB1perp}.

  \medskip
  \noindent
  \textbf{Step 7: Conclusion.} To obtain the desired conclusion, we
  only have to apply Theorem \ref{thm:GMM} with $\lambda_1=\lambda_0$
  being the spectral gap of $\mathbf{L}$ while $0 <
  \lambda_\star=\lambda_2$ is any positive number strictly smaller
  than $\min(\lambda_0,\lambda_3)$ where $\lambda_3$ is constructed in
  Step 3. We just recall here that, if $\lim_{i \to \infty} a_i =
  +\infty$, the number $\lambda_3$ is any arbitrary real number
  strictly larger than $\lambda_0$, and thus in this case any
  $\lambda_\star \in (0,\lambda_0)$ would yield the conclusion.
\end{proof}

\section{Exponential convergence for the Becker-Döring equations}
\label{sec:nonlinear}

\subsection{Local exponential convergence}

In this section we prove a local version of Theorem
\ref{thm:main-intro*}.

\begin{thm}[\textbf{Local exponential convergence for the
    Becker-D\"oring equations}]
  \label{thm:exp-local-BD}
  Assume \eqref{eq:hyp-existence} and \eqref{eq:hyp-JN}. Let $c =
  (c_i)_i$ be a nonzero subcritical solution to equation \eqref{eq:BD}
  (i.e., with density $\varrho < \varrho_\mathrm{s}$) with initial condition
  $c(0)$ such that \eqref{eq:initial-exp-moment-finite} holds for some
  $\nu > 0$. Take $z > 0$ satisfying \eqref{eq:z-def}.

  Then there exist some $0 < \eta < \nu$, some $C, \epsilon > 0$ and
  some $\lambda_\star > 0$ (all explicit) such that if
  $$\sum_{i=1}^\infty \exp\left(\eta\,i\right)
  |c_i(0) - \Q_i| < \epsilon$$
  then
  \begin{equation}
    \label{eq:exp-local-BD}
    \sum_{i=1}^\infty \exp\left(\eta\,i\right) |c_i(t) - \Q_i|
    \leq
    C \exp\left(-\lambda_\star t\right)
    \fa t \geq 0.
  \end{equation}
  In addition, $\lambda_\star$ can be taken equal to $\lambda_0$ (the
  size of the spectral gap of the linear operator $\L$) if $\lim_{i
    \to +\infty} a_i = +\infty$.
\end{thm}
We remark that all the constants in the above result can be explicitly
estimated. The proof of the above local convergence results relies on
two crucial estimates. The first one was proved in
Prop. \ref{lem:Gamma-bound-DCF-l1}, a bound of the nonlinear term
$\Gamma(h,h)$; the second one is a uniform--in--time bound of
exponential moments for the Becker-D\"oring equations, which we take
from \cite{JN03}:

\begin{prp}
  \label{prop:moments}
  Assume \eqref{eq:hyp-existence} and the conditions
  \eqref{eq:hyp-JN}. Let $c = (c_i)_i$ be a nonzero subcritical
  solution to equation \eqref{eq:BD} with initial condition $c(0)$
  such that \eqref{eq:initial-exp-moment-finite} holds for some $\nu >
  0$. Then there exists $0 < \overline{\nu} < \nu$ and $K_1 > 0$ such
  that
  \begin{equation}
    \label{uniform}
    \sum_{j=1}^\infty \exp\left(\overline{\nu}\, j\right) c_j(t) \leq K_1,
  \end{equation}
  where $c=(c_i(t))_i$ is the unique solution to \eqref{eq:BD} with
  initial datum $c(0)=(c_i(0))_{i\geq1}.$
\end{prp}

We notice that \eqref{eq:hyp-existence} and \eqref{eq:hyp-JN},
together with the definition of $Q_i$ through \eqref{eq:Qi}, imply
conditions (H1)--(H4) of \cite{JN03}. In particular, since
$z_\mathrm{s}$ is defined as the radius of convergence of the series
$\sum z^i i Q_i$, it is clear that
\begin{equation*}
  \lim_{i \to +\infty} Q_i^{1/i} = \frac{1}{\zs},
\end{equation*}
which is (H3) of \cite{JN03}.

We are ready then to prove Theorem \ref{thm:exp-local-BD}:
\begin{proof}[Proof of Theorem \ref{thm:exp-local-BD}]
  Let $c$ be a solution of \eqref{eq:BD} and let $h = (h_i)$ be the
  fluctuation around the equilibrium, defined as in
  (\ref{eq:def-hi}). Since Proposition \ref{prop:moments} holds under
  our conditions, we can find $0 < \overline{\nu} < \nu$ such that
  \eqref{uniform} holds. We take any $\eta < \min\{\overline{\nu},
  \frac{1}{2} \log \frac{\zs}{z} \}$, and consider $X$ the vector
  space defined in \eqref{eq:def-X}. Note that Proposition
  \ref{lem:Gamma-bound-DCF-l1} and Theorem
  \ref{thm:spectral-gap-extension} are applicable under these
  conditions, and consider the quantity $\lambda_\star$ appearing in
  Theorem \ref{thm:spectral-gap-extension}.

  The fluctuation $h$ satisfies equation (\ref{eq:hi-equation}) in
  $X$:
  \begin{equation*}
    \dfrac{\d }{\d t} h
    =
    \Lambda[h] + \Gamma(h,h),
  \end{equation*}
  so, if $(S_t)_{t \geq 0}$ denotes the evolution semigroup generated
  by $\Lambda$, then
  \begin{equation*}
    h(t) = S_t h(0) + \int_0^t S_{t-s} \Gamma(h(s),h(s)) \ds.
  \end{equation*}
  Recall that $\sum_{i \geq 1}ih_i(t)\Q_i=0=\sum_{i \geq 1}
  i\Gamma_i(h(t),h(t))\Q_i$ for any $t \geq 0$ so that, according to
  Proposition \ref{lem:Gamma-bound-DCF-l1} and Theorem
  \ref{thm:spectral-gap-extension}, we have (for some constant $C >
  0$),
  \begin{multline}
    \label{eq:nl11}
    \|h(t)\|_X
    \leq
    \|S_t h(0)\|_X
    +
    \int_0^t \| S_{t-s} \Gamma(h(s),h(s)) \|_X \ds
    \\
    \leq
    C \|h(0)\|_X \exp(-\lambda_\star  t)
    +
    C \exp(-\lambda_\star  t)
    \int_0^t \exp( \lambda_\star s) \|\Gamma(h(s),h(s)) \|_X \ds
    \\
    \leq
    C \|h(0)\|_X \exp(-\lambda_\star  t)
    +
    C \exp(-\lambda_\star t)
    \int_0^t \exp( \lambda_\star  s) \|h(s)\|_X \|h(s)\|_{X_1} \ds
  \end{multline}
  where we recall that $X_1$ is defined by
  \eqref{eq:def-Xlambda}.  For any $\delta \in (0,
  \overline{\nu}-\eta)$, we have from Cauchy-Schwarz's inequality that
  \begin{multline*}
    \|h(t)\|_{X_1}
    = \sum_i |h_i(t)| i  \Q_i \exp\left(\eta\,i\right)\\
    \leq  \left(
      \sum_{i=1}^{\infty} |h_i(t)| i   \Q_i \exp\left(\eta\,i - \delta\,i\right)
    \right)^{1/2}
    \left(
      \sum_{i=1}^{\infty}  |h_i(t)| i  \Q_i \exp\left(\eta\,i + \delta\,i\right)
    \right)^{1/2}.
  \end{multline*}
  Now, using that
  \begin{equation*}
    i\exp\left(-\delta i\right) \leq C,
    \qquad
    i \exp\left(\eta\,i + \delta\,i\right)
    \leq C \exp\left(\overline{\nu} i\right)
    \fa i \geq 1
  \end{equation*}
  we deduce that
  \begin{equation*}
    \|h(t)\|_{X_1}
    \leq
    C \|h(t)\|_{X}^{1/2}   \bigg(\sum_{i=1}^{\infty}
    |h_i(t)| \Q_i \exp\left(\overline{\nu} i\right) \bigg)^{1/2}
    \leq
    C K_1^{1/2} \|h(t)\|_{X}^{1/2}
    \fa t \geq 0.
  \end{equation*}
  Using this in eq. \eqref{eq:nl11} we have
  \begin{equation*}
    \|h(t)\|_X
    \leq
    C \|h(0)\|_X \exp(-\lambda_\star t)
    +
    C \exp(-\lambda_\star  t)
    \int_0^t \exp(\lambda_\star s) \|h(s)\|_X^{3/2} \ds.
  \end{equation*}
  Then, the quantity $u(t) := \exp(\lambda_\star t) \|h(t)\|_X$ satisfies
  \begin{equation*}
    u(t) \leq C \|h(0)\|_X + C \int_0^t u(s)^{3/2} \d s.
  \end{equation*}
  We deduce then the conclusion from Gronwall's inequality.
\end{proof}

\subsection{Global exponential convergence}
\label{sec:global-BD}

In order to prove Theorem \ref{thm:main-intro*} we complement the
local result given in Theorem \ref{thm:exp-local-BD} with the
following global result from \cite{JN03}:
\begin{prp}
  \label{prop:JN03}
  Assume \eqref{eq:hyp-existence} and the conditions
  \eqref{eq:hyp-JN}. Let $c = (c_i)_i$ be a nonzero subcritical
  solution to equation \eqref{eq:BD} with initial condition $c(0)$
  such that \eqref{eq:initial-exp-moment-finite} holds for some $\nu >
  0$, and take $z > 0$ satisfying \eqref{eq:z-def} as usual.

  Then, there exist two explicit constants $C_1>0$ and $\kappa_1>0$
  such that
  \begin{equation}
    \label{eq:JN03}
    \sum_{j \geq 1} j\left|c_j(t)-\Q_j\right|
    \leq C_1 \exp(-\kappa_1 t^{\frac{1}{3}})
    \qquad \forall t \geq 0.
  \end{equation}
\end{prp}

\begin{proof}[Proof of Theorem \ref{thm:main-intro*}]
  Let $h = (h_i)$ be defined as in (\ref{eq:def-hi}). Let us fix the
  same $\eta \in (0, \overline{\nu})$ as in Theorem
  \ref{thm:exp-local-BD} and denote by $X$ the $\ell^1$ space with
  weight $e^{\eta i}$ that was defined in \eqref{eq:def-X}. The idea
  is to use Proposition \ref{prop:JN03} to show that at a certain time
  $t_0$ we are close enough to the equilibrium to use Theorem
  \ref{thm:exp-local-BD}. To do this we use a simple interpolation
  argument to estimate the norm
  $$ \sum_{i \geq 1}\left|c_i(t) - \Q_i \right|\exp\left(\eta i\right) = \|h(t)\|_X$$
  in terms of $\sum_{i \geq 1} i \left|c_i(t)-\Q_i\right|$ and a
  slightly stronger norm. Namely, for any $\delta \in
  (0,\overline{\nu}-\eta)$, a simple interpolation argument yields
    \begin{equation*}\begin{split}
      \|h(t)\|_X&=\sum_{i \geq 1}|h_i(t)|\Q_i \exp\left(\eta i\right)\leq \left(\sum_{i \geq 1} i|h_i(t)|\Q_i \right)^{\frac{\delta}{\delta+\eta}}
      \left(\sum_{i \geq 1}  i^{-1}|h_i(t)|\Q_i \exp\left((\eta+\delta)i\right)\right)^{\frac{\eta}{\eta+\delta}}\\
      &\leq \left(\sum_{i \geq 1} i |h_i(t)|\Q_i \right)^{\frac{\delta}{\delta+\eta}} \left(\sum_{i \geq 1} |h_i(t)|\Q_i \exp\left(\overline{\nu}\,i\right)\right)^{\frac{\eta}{\eta+\delta}}.
    \end{split}\end{equation*}
  Therefore, according to \eqref{uniform} and \eqref{eq:JN03}, one
  obtains
  \begin{equation}
    \label{eq:global-decay}
    \|h(t)\|_X \leq
    C_1^{\frac{\delta}{\delta+\eta}}
    K_1^{\frac{\eta}{\eta+\delta}}
    \exp\left(-{\frac{\kappa_1\delta}{\delta+\eta}}
      t^{\frac{1}{3}}\right)
    \qquad \forall t \geq 0.
  \end{equation}
  Consider the number $\varepsilon >0$ given by Theorem
  \ref{thm:exp-local-BD}. Then, there exists an explicit $t_0 >0$ such that
  $$\|h(t_0)\|_X \leq \varepsilon$$
  and one deduces from \eqref{eq:exp-local-BD} starting from $t=t_0$
  that
  $$\|h(t)\|_X \leq C_\mu \|h(t_0)\|_X \exp(-\lambda_\star t)
  \qquad \forall t \geq t_0.$$
  Together with \eqref{eq:global-decay}
  for $t < t_0$, this concludes the proof.
\end{proof}
 \appendix

\section{Discrete Hardy's inequalities}
\label{sec:Hardy}
\setcounter{equation}{0}
\renewcommand{\theequation}{A.\arabic{equation}}
The following discrete version of Hardy's inequality is used in the
proof of Theorem \ref{thm:spectral-gap}:
\begin{thmA}
  Consider two sequences of positive numbers $(\mu_i)_i$ and
  $(\nu_i)_i$. Then, the following are equivalent:
  \begin{enumerate}
  \item There exists a finite constant $A \geq 0$ such that
    \begin{equation}
      \label{Hardy}
      \sum_{i=1}^\infty \mu_i\left(\sum_{j=1}^i f_j\right)^2 \leq A \sum_{i=1}^\infty \nu_i f_i^2
    \end{equation}
    for any sequence $f=(f_i)_i$.
  \item The following holds
    $$B=\sup_{k \geq 1}\left(\sum_{j=k}^\infty \mu_j\right)\,\left(\sum_{i=1}^k \dfrac{1}{\nu_i}\right) < \infty.$$
  \end{enumerate}
  If these equivalent propositions hold true, then $B \leq A \leq 4B.$
\end{thmA}

\begin{remA}
  The above discrete version of Hardy's inequality is not new
  \cite{Miclo} and we give a proof for the sake of completeness. It is
  a simple adaptation of the proof of the continuous version of
  Hardy's inequality, originally due to Muckenhoupt \cite{muck}, that
  can be found in \cite[Chapitre 6]{Ane}.
\end{remA}

\begin{proof}
  Assume that there exists some finite $A >0$ such that \eqref{Hardy}
  hold true for any $f=(f_i)_i$ and let us prove that $A \geq B.$ For
  any fixed $k \geq 1$, set
  $$f_i=\begin{cases} \dfrac{1}{\nu_i} \qquad \text{ for } i \leq k\\
    0 \qquad \text{ for } i > k.\end{cases}$$ Then, one recognizes
  easily that $\sum_{i=1}^\infty \nu_i f_i^2=\sum_{i=1}^k
  \frac{1}{\nu_i}$ while
  $$\sum_{i=1}^\infty \mu_i\left(\sum_{j=1}^i f_j\right)^2 \geq \left(\sum_{i=k}^\infty \mu_i\right)\,\left(\sum_{j=1}^k \dfrac{1}{\nu_j}\right)^2.$$
  According to \eqref{Hardy} one gets then
  $$A \sum_{i=1}^k \dfrac{1}{\nu_i} \geq \left(\sum_{i=k}^\infty \mu_i\right)\,\left(\sum_{j=1}^k \dfrac{1}{\nu_j}\right)^2$$
  or equivalently $A \geq \left(\sum_{i=k}^\infty
    \mu_i\right)\,\left(\sum_{j=1}^k \frac{1}{\nu_j}\right)$ for any
  $k \geq 1.$ This proves that $B \leq A.$

  Conversely, let us assume that $B < \infty$ and let us prove that $A
  \leq 4B.$ Set
$$\Upsilon_i=\sum_{j=1}^i \dfrac{1}{\nu_j}\,, \qquad \Gamma_i=\sum_{j=i}^\infty \mu_j, \qquad \gamma_i=\sqrt{\Gamma_i} \qquad \text{ and }\qquad \beta_i=\sqrt{\Upsilon_i}.$$
Let $f=(f_i)_i$ be a given sequence. Thanks to Cauchy-Schwarz inequality
$$\left(\sum_{j=1}^i f_j\right)^2 \leq \left(\sum_{j=1}^i f_j^2\, \nu_j \beta_j\right)\,\left(\sum_{j=1}^i \dfrac{1}{\nu_j \beta_j}\right)$$
Now, since $\frac{1}{\nu_j}=\Upsilon_j-\Upsilon_{j-1}$, one gets
$$\dfrac{1}{\nu_j \beta_j}=\dfrac{1}{\nu_j \sqrt{\Upsilon_j}}=\dfrac{\Upsilon_j-\Upsilon_{j-1}}{\sqrt{\Upsilon_j}} \qquad \forall j \geq 1.$$
Using the general inequality
\begin{equation}\label{sqrtXY}\dfrac{X-Y}{\sqrt{X}} \leq 2\left(\sqrt{X}-\sqrt{Y}\right) \qquad \forall X \geq Y > 0\end{equation}
we get that
$$\left(\sum_{j=1}^i \dfrac{1}{\nu_j \beta_j}\right) \leq 2\sum_{j=1}^i \left(\sqrt{\Upsilon_j}-\sqrt{\Upsilon_{j-1}}\right)=2\sqrt{\Upsilon_i}=2\beta_i$$
Therefore, the left-hand-side of \eqref{Hardy} can be estimated as follows:
$$\sum_{i=1}^\infty \mu_i\left(\sum_{j=1}^i f_j\right)^2 \leq 2\sum_{i=1}^\infty \mu_i\,\beta_i \left(\sum_{j=1}^i f_j^2\, \nu_j\,\beta_j\right)=2\sum_{j=1}^\infty f_j^2 \nu_j\,\beta_j \left(\sum_{i=j}^\infty \mu_i\,\beta_i\right).$$
Then, one sees that to prove our claim, it suffices to prove that
\begin{equation}\label{2B}\beta_j\,\left(\sum_{i=j}^\infty \mu_i\,\beta_i\right) \leq 2 B \qquad \forall j \geq 1.\end{equation}
One notices that, by definition of $B$, $\beta_i  \gamma_i\leq \sqrt{B}$ for any $i \geq 1.$ Therefore,
$$\sum_{i=j}^\infty \mu_i\,\beta_i \leq \sqrt{B} \sum_{i=j}^\infty \dfrac{\mu_i}{\gamma_i}=\sqrt{B}\sum_{i=j}^\infty \dfrac{\mu_i}{\sqrt{\Gamma_i}}.$$
Now, since $\Gamma_i-\Gamma_{i+1}=\mu_i$ one has, from \eqref{sqrtXY}:
$$\dfrac{\mu_i}{\sqrt{\Gamma_i}}=\dfrac{\Gamma_i-\Gamma_{i+1}}{\sqrt{\Gamma_i}} \leq 2 \left(\sqrt{\Gamma_i}-\sqrt{\Gamma_{i+1}}\right)$$
and the above inequality reads
$$\sum_{i=j}^\infty \mu_i\,\beta_i \leq 2\sqrt{B} \sum_{i=j}^\infty \left(\sqrt{\Gamma_i}-\sqrt{\Gamma_{i+1}}\right)=2\sqrt{B}\sqrt{\Gamma_j}=2\sqrt{B}\gamma_j \qquad \forall j \geq 1.$$
Finally, by definition of $B$, we get that
$$\beta_j \,\left(\sum_{i=j}^\infty \mu_i\,\beta_i\right)  \leq 2\sqrt{B} \beta_j \gamma_j \leq 2B \qquad \forall j \geq 1$$
which gives \eqref{2B} and achieves the proof.
\end{proof}

\section{Some estimates of special functions}
\setcounter{equation}{0}
\renewcommand{\theequation}{B.\arabic{equation}}
We collect here some technical estimates for special functions which turn useful for estimating the spectral gap size near the critical density (see Section \ref{sub:critical}).

\begin{lemB}
  \label{lem:incomplete-gamma-estimate}
  Take $0 < \mu \leq 1$. For some quantity $C_\mu$ depending only on
  $\mu$,
  \begin{equation}
    \label{eq:incomplete-gamma-estimate}
    \int_x^\infty \exp(-y^\mu) \dy
    \leq
    C_\mu \exp(-x^\mu) x^{1-\mu}
    \fa x \geq 1.
  \end{equation}
\end{lemB}

\begin{proof}
  Through the change of variables $u = y^\mu - x^\mu$ we obtain
  \begin{equation*}
    \int_x^\infty \exp(-y^\mu) \dy
    =
    \frac{1}{\mu} \exp(-x^\mu) x^{1-\mu}
    \int_0^\infty \exp(-u) \left(
      \frac{u}{x^\mu} + 1
    \right)^{\frac{1-\mu}{\mu}}
    \du
  \end{equation*}
  so the result holds with $C_\mu = \frac{1}{\mu} \int_0^\infty
  \exp(-u) \left( u + 1 \right)^{\frac{1-\mu}{\mu}}\du$.
\end{proof}

\begin{lemB}
  \label{lem:incomplete-gamma-estimate-2}
  Take $0 < \mu \leq 1$ and $0 \leq \alpha < 1$. For some quantity
  $C_{\mu,\alpha}$ depending only on $\mu$ and $\alpha$,
  \begin{equation}
    \label{eq:incomplete-gamma-estimate-2}
    \int_0^x \exp(y^\mu) y^{-\alpha} \dy
    \leq
    C_{\mu,\alpha} \, \exp(x^\mu) x^{1-\mu-\alpha}
    \fa x \geq 1.
  \end{equation}
\end{lemB}

\begin{proof}
  We first break the integral in the intervals $(0,x/2]$ and
  $(x/2,x)$: for the first part,
  \begin{equation}
    \label{eq:ig1}
    \int_0^{\frac{x}{2}} \exp(y^\mu) y^{-\alpha} \dy
    \leq
    \exp((x/2)^\mu) \int_0^{\frac{x}{2}} y^{-\alpha} \dy
    =
    \frac{1}{1-\alpha}
 \exp((x/2)^\mu) \left(\frac{x}{2}\right)^{1-\alpha}
    \leq
    A_{\mu,\alpha} \exp(x^\mu) x^{1-\mu-\alpha}
  \end{equation}
  with
  \begin{equation*}
    A_{\mu,\alpha}
    :=  \frac{2^{\alpha-1}}{1-\alpha}
    \max_{x \geq 1} \left\{
      \exp({(x/2)^\mu - x^\mu}) x^{\mu}
    \right\}.
  \end{equation*}
  For the second part, through the change of variables $u = x^\mu -
  y^\mu$ we obtain
  \begin{multline}
    \label{eq:ig2}
    \int_{\frac{x}{2}}^x \exp({y^\mu}) y^{-\alpha} \dy
    =
    \frac{1}{\mu} \exp({-x^\mu}) x^{1-\mu-\alpha}
    \int_0^{\frac{x^\mu}{2}} \exp({-u}) \left(
      1 - \frac{u}{x^\mu}
    \right)^{\frac{1-\mu-\alpha}{\mu}}
    \du
    \\
    \leq
    B_{\mu,\alpha}
    \exp({-x^\mu}) x^{1-\mu-\alpha}
    \int_0^{\frac{x^\mu}{2}} \exp({-u})
    \du
    \leq
    B_{\mu,\alpha}
\exp({-x^\mu}) x^{1-\mu-\alpha}
  \end{multline}
  with
  \begin{equation*}
    B_{\mu,\alpha} :=
    \frac{1}{\mu}
    \max\left\{1, 2^\frac{\mu}{1-\mu-\alpha} \right\}
  \end{equation*}
  since $\frac{1}{2} \leq 1 - \frac{u}{x^\mu} \leq 1$ on the region of
  integration. From \eqref{eq:ig1} and \eqref{eq:ig2} we conclude
  with $C_{\mu,\alpha} := \max\{A_{\mu,\alpha}, B_{\mu,\alpha} \}$.
\end{proof}

We end this Appendix with an asymptotic result for special expansion series:
\begin{lemB}
  \label{lem:exp-sum-asymptotics}
  Take $\mu > 0$ and a sequence $(g_k)_k$ such that, for some $C_0 > 0$,
  \begin{equation}
    \label{eq:hk-increment-condition}
    g_{k+1}-g_{k} \sim C_0 k^{\mu - 1}
    \quad \text{ as } k \to +\infty.
  \end{equation}
  Then
  \begin{equation}
    \label{eq:leftsum-asymptotics}
    \sum_{j=1}^k \exp({g_j}) j^\nu
    \sim
    \frac{1}{C_0}
\exp({g_k}) k^{\nu + 1 - \mu}
    \quad \text{ as } k \to +\infty
  \end{equation}
  and
  \begin{equation}
    \label{eq:rightsum-asymptotics}
    \sum_{j=k}^\infty \exp({-g_j}) j^\nu
    \sim
    \frac{1}{C_0}
\exp({-g_k}) k^{\nu + 1 - \mu}
    \quad \text{ as } k \to +\infty.
  \end{equation}
\end{lemB}

\begin{remB}
  \label{rem:exp-sum}
  One easily deduces that, in the conditions of the above lemma, there
  is a constant $C$ which depends only on the sequence $(g_k)_k$ and
  $\nu$ such that
  \begin{gather*}
    \sum_{j=1}^k \exp({g_j}) j^\nu \leq C \exp({g_k}) k^{\nu + 1 - \mu} \quad \text{ and } \quad
    \sum_{j=k}^\infty \exp({-g_j}) j^\nu
    \leq C \exp({-g_k}) k^{\nu + 1 - \mu}
  \end{gather*}
  for all $k \geq 1$.
\end{remB}

\begin{proof}
  We prove \eqref{eq:leftsum-asymptotics} by using the Stolz-Ces\`aro
  theorem (which is a discrete analogue of l'H\^opital's
  rule). Namely, setting
  $$U_k=\exp({g_k}) k^{\nu + 1 - \mu} \qquad \text{ and }  \qquad V_k=\sum_{j=1}^k \exp({g_j}) j^\nu \qquad (k \geq 1),$$
  in order to show that
  \begin{equation*}
\lim_{k \to \infty} \dfrac{U_k}{V_k}=C_0,
  \end{equation*}
  we instead prove that
\begin{equation}\label{cesaro}\lim_{k \to \infty} \dfrac{U_{k+1}-U_k}{V_{k+1}-V_k}=C_0.\end{equation}
One has  \begin{multline*}
 \dfrac{U_{k+1}-U_k}{V_{k+1}-V_k}
    =
    \frac{(k+1)^{\nu + 1 - \mu} - \exp({g_{k}-g_{k+1}}) k^{\nu + 1 - \mu}}
    {(k+1)^\nu}
    \\
    =
    \frac{\big(
        (k+1)^{\nu + 1 - \mu} - k^{\nu + 1 - \mu}
      \big)
      +  k^{\nu + 1 - \mu}\left(1 - \exp({g_{k}-g_{k+1}})\right)
    }
    {(k+1)^\nu}.
  \end{multline*}
Since
  \begin{equation*}
    \frac{ \big|
      (k+1)^{\nu + 1 - \mu} - k^{\nu + 1 - \mu}
      \big|
    }
    {(k+1)^\nu}
    \leq
    \frac{|\nu + 1 - \mu| \min\{(k+1)^{\nu - \mu}, k^{\nu - \mu}\}}
    {(k+1)^\nu}
    \sim
    |\nu + 1 - \mu|\, k^{-\mu}
    \to 0
  \end{equation*}
  as $k \to +\infty$ while
  \begin{equation*}
  \lim_{k \to \infty}  \frac{k^{\nu + 1 - \mu}\left(1 - \exp({g_{k}-g_{k+1}})\right)}
    {(k+1)^\nu}
    =
\lim_{k \to \infty}    \frac{ k^{\nu + 1 - \mu} ({g_{k+1}-g_{k}})}
    {(k+1)^\nu} \frac{1 - \exp({g_{k}-g_{k+1}})}{g_{k+1}-g_{k}}
    =
    C_0,
  \end{equation*}
  using $\left(1 - \exp({g_{k}-g_{k+1}})\right) / (g_{k+1}-g_{k}) \to
  1$ together with \eqref{eq:hk-increment-condition}, we obtain
  \eqref{cesaro}. We apply the Stolz-Ces\`aro theorem and obtain
  \eqref{eq:leftsum-asymptotics} (notice that, since
  \eqref{eq:hk-increment-condition} implies that $g_k \sim (C_0 / \mu)
  k^\mu$, the sum $\sum_{j=1}^k \exp({g_j}) j^\nu$ diverges and hence
  the hypotheses of the theorem are satisfied).  Equation
  \eqref{eq:rightsum-asymptotics} can be proved by a very similar
  calculation.
\end{proof}

\bibliographystyle{plain}


\end{document}